\documentclass{article}
\usepackage[utf8]{inputenc}
\usepackage{style_defaults}
\usepackage{amsmath}
\usepackage{authblk}
\usepackage{macros_essential}
\usepackage{algorithmic, algorithm}

\usepackage{macros_essential}
\usepackage{style_defaults}
\usepackage{algorithmic, algorithm}
\renewcommand{\algorithmiccomment}[1]{$\triangleright$ \textit{#1}}
\usepackage{wrapfig}

\usepackage{amssymb}
\usepackage[backref,colorlinks,citecolor=blue,bookmarks=true]{hyperref}
\usepackage[letterpaper,margin=1in]{geometry}

\usepackage[nameinlink]{cleveref}
\crefname{ineq}{inequality}{inequalities}
\creflabelformat{ineq}{#2{\upshape(#1)}#3}
\title{Batching of Tasks by Users of Pseudonymous Forums:\\ Anonymity Compromise and Protection}

\author{Alexander Goldberg}
\author{Giulia Fanti}
\author{Nihar B. Shah}
\affil{Carnegie Mellon University}
\affil{\texttt {\{akgoldbe,gfanti,nihars\}@andrew.cmu.edu}}
\date{}

\begin{document}
\maketitle

\begin{abstract}
    There are a number of forums where people participate under pseudonyms. One example is peer review, where the identity of reviewers for any paper is confidential. When participating in these forums, people frequently engage in ``batching'': executing multiple related tasks (e.g., commenting on multiple papers) at nearly the same time. Our empirical analysis shows that batching is common in two applications we consider -- peer review and Wikipedia edits. 
    In this paper, we identify and address the risk of deanonymization arising from linking batched tasks. To protect against linkage attacks, we take the approach of adding delay to the posting time of batched tasks. We first show that under some natural assumptions, no delay mechanism can provide a meaningful differential privacy guarantee. We therefore propose a ``one-sided'' formulation of differential privacy for protecting against linkage attacks. We design a mechanism that adds zero-inflated uniform delay to events and show it can preserve privacy. We prove that this noise distribution is in fact optimal in minimizing expected delay among mechanisms adding independent noise to each event, thereby establishing the Pareto frontier of the trade-off between the expected delay for batched and unbatched events. Finally, we conduct a series of experiments on Wikipedia and Bitcoin data that corroborate the practical utility of our algorithm in obfuscating batching without introducing onerous delay to a system.
    
\end{abstract}

\section{Introduction}
\label{sec:intro}

\begin{wrapfigure}[20]{r}{.34\textwidth}
    \centering
    \vspace{-45pt}
    \includegraphics[width=.34\textwidth]{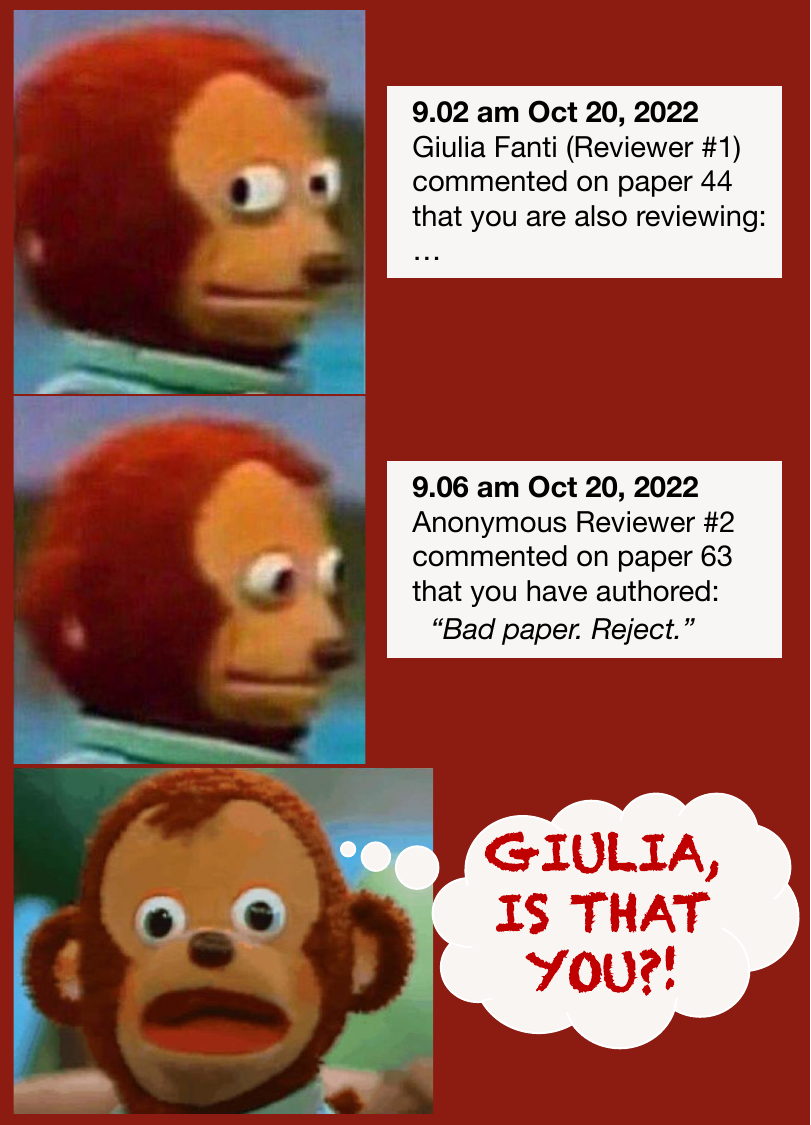}
    \caption{Cartoon illustration of reviewer de-anonymization due to batching.}
    \label{FigMeme}
\end{wrapfigure}

In a number of applications where anonymity is critical, users act under pseudonyms to preserve their privacy. For instance, in scientific peer review using online forums like \href{https://openreview.net/}{OpenReview.net}, reviewers make comments on papers that are publicly viewable. Reviewers (and meta-reviewers) who have been assigned multiple papers operate under different pseudonyms across their papers to remain anonymous. Other examples of publicly visible tasks where users operate under pseudonyms include Wikipedia editing and cryptocurrency transactions.

In many settings, it is common for users to engage in {\it batching} --- the completion of several similar tasks at the same time. Batching occurs both due to natural bursts in activity (e.g., a person visits a website and makes many comments at once) or as a productivity strategy used to streamline work. Indeed, both academic studies~\cite{kushlev_email_stress_2015, mark2016email, blank_email_emotional_2020} and popular media~\cite{monday_blog_batching, forbes_batching, nytimes_email} recommend performing tasks like responding to emails in batches in order to improve efficiency and reduce work-related stress.

In peer-review forums such as computer science conferences, reviewers and meta-reviewers are often assigned multiple papers. We find empirically that reviewers and meta-reviewers are highly likely to batch their comments and/or reviews. Specifically, we analyze data from a top Computer Science conference\footnote{Name redacted for privacy.} with thousands of papers, reviewers, and discussion comments. We find that when reviewers and meta-reviewers comment on multiple papers, they have a 30.10\% chance of batching their comments within $5$ minutes of one other. In comparison, any randomly chosen pair of reviewers and meta-reviewers had only a 0.66\% chance of making comments on different papers within $5$ minutes of each other. 

While batching is normal human behavior, it introduces a risk of deanonymization in peer-review settings.\footnote{This outcome is bad for a review system that needs a lot of interaction with the authors, but not for conferences where this is not expected nor allowed, like AAAI and IJCAI. The conference we analyzed was not on OpenReview.net but on a different conference management platform that does not make discussions public and has only a single-shot interaction between reviewers and authors (via a "rebuttal"). It is of interest to see an analogous analysis on conferences on OpenReview.net, but we do not have access to such data.}   For example, in many open peer-review settings, 
comments are publicly posted. Furthermore, many conferences have policies that (meta-)reviewers for any paper know the identities of other (meta-)reviewers on that paper. Now, when a (meta-)reviewer batches their comments, an author may observe that two comments are generated at nearly the same time on their own paper and on another paper. The author can then link the identity of this anonymous (meta-)reviewer on their own paper to a (meta-)reviewer on the other paper. If the author knows the identity of the (meta-)reviewers on the other paper---for instance, if the author is the meta-reviewer or another reviewer for that paper---this can uncover the identity of the (meta-)reviewer of their own paper. See Figure~\ref{FigMeme} for a cartoon illustration.

A back-of-the-envelope calculation based on our aforementioned measurements in peer review suggests that if an author has a uniform prior over $10$ possible (meta-)reviewers of their paper, then after observing a comment posted on their own paper within $5$ minutes of another comment from one of these (meta-)reviewers on another paper, their posterior probability that this (meta-)reviewer made the comment increases to $\frac{0.301}{0.301 + 9(0.0066)} = 83.51\%$ as compared to the prior of $10\%$. Thus, the linking of (meta-)reviewers across papers using batched comments can undermine the anonymity of the peer review process.

Similar privacy risks due to batching arise in many systems where users generate publicly logged events under pseudonyms. For instance:
\begin{itemize}[leftmargin=*]
    \item \textit{Inferring the identity of editors on Wikipedia articles.} Wikipedia provides public edit histories of articles. 
   While edit history is public, Wikipedia users are known to maintain their anonymity for a variety of important reasons. For instance, one study of Wikipedia editors who use the anonymity network Tor found that editors are concerned about their privacy due to risks like ``\emph{threats of surveillance, violence, harassment, opportunity loss, reputation loss, and fear for loved ones.}''~\cite{wiki_privacy} These risks are especially acute for already marginalized groups like women and ethnic minorities. Thus, the study concludes that ``\emph{open collaboration
    communities must go beyond attracting participants, to
    develop social and technical arrangements that support
    contributors’ needs for privacy.}'' 
    
    In order to address these privacy concerns, Wikipedia's terms of service explicitly allow for the use of a pseudonymous alternate account: \emph{``A person editing an article that is highly controversial within their family, social or professional circle, and whose Wikipedia identity is known within that circle, or traceable to their real-world identity, may wish to use an alternative account to avoid real-world consequences from their editing or other Wikipedia actions in that area}.''~\cite{wiki_TOS_pseudonyms} However, as in the peer review example, batched timing of article revisions can enable linkage of the second account to a known primary account. In practice, the batching of edits is ubiquitous on Wikipedia; our analysis of publicly logged Wikipedia article revisions shows that over $50\%$ of all edits are made within $5$ minutes of an edit from the same user on a different article. This common editing behavior may undermine the privacy of users employing a second account to preserve their anonymity.
    \item \textit{Clustering crypto-currency transactions on a public blockchain.} In cryptocurrencies like Bitcoin, users' transaction histories are recorded on a public blockchain where a person can send or receive currency to an associated public key, which acts as a pseudonym. Users can have multiple addresses, each containing its own funds and identified by a different public key \cite{bitcoin_address}. A transaction can (and often does) draw funds from multiple input addresses, particularly if no single address contains sufficient funds for a given transaction \cite{tx-format}. 
    However, a common heuristic used in practice is to link multiple input addresses to a single transaction to the same user \cite{common-input,bitcoin_man_with_no_names_2016}. Hence, users who wish to preserve their privacy can separate inputs from different addresses into different transactions to obfuscate the linkage between transactions from the same person \cite{common-input}.\footnote{There exist other cryptographic solutions (e.g., CoinJoin) that leak more information in exchange for cost benefits compared to generating multiple transactions \cite{coinjoin}.}
    However, if a user batches these transactions in time across addresses, an adversary may use this timing (along with other signals) to still link together their multiple addresses. Linking pseudonyms together is a common first step in a full deanonymization attack. For instance, attacks on Bitcoin transactions begin by leveraging a user's ``idioms of use'' to cluster together addresses likely belonging to the same person~\cite{bitcoin_man_with_no_names_2016}. The attacker then leverages a single known link to a real-world identity to de-anonymize the entire cluster. 
\end{itemize}

These scenarios motivate the need for defenses against timing-based linkage attacks that exploit the batching of tasks by people. There is already extensive literature on privacy-preserving data release in various settings. However, there are a number of strict constraints in our setting that prevent these methods from being applicable. A common approach to preserving privacy is to introduce fake events to obscure patterns among the real events. However, in all three applications --- peer review, Wikipedia, and cryptocurrency --- generating fake events is highly undesirable or impractical, and withholding events indefinitely is also not possible.  In the setting of commenting in peer review, it is undesirable to generate fake comments, as this would require giving made-up feedback to paper authors. Similarly, in Wikipedia, adding fake edits to articles can undermine the quality and legitimacy of the content. For cryptocurrencies, introducing dummy transactions would introduce additional financial cost, causing undesirable overhead. Furthermore, transactions include the amount of currency sent, so dummy transactions would require a sender to transfer actual funds just to preserve privacy. Instead, our approach is to design \emph{delay mechanisms} that introduce random delays to the time at which events are posted on the platform (without the use of any dummy data) to preserve privacy. Thus, the mechanism will trade off privacy for additional delay in the system.  

\medskip{\bf Our contributions.} In this work, we introduce the problem of anonymity compromise due to the batching of tasks in pseudonymous forums and then propose defenses. Our primary contributions are:
\begin{itemize}
    \item We identify the problem of deanonymization risk due to the batching of tasks by users of pseudonymous online forums. By analyzing data from an actual peer-reviewed conference, we demonstrate that a simple attack using the timing of comments on an online forum can link anonymous (meta)-reviewer's identities, increasing their certainty about a specific (meta)-reviewer's identity to 83\% from a prior of 10\%.
    In analysis of Wikipedia article revisions, we show that  batching of tasks on Wikipedia makes it possible to link editors across articles with an accuracy of 85\% 
    based only on the timing of their revisions.
    \item We formulate the problem of trading off privacy and delay in pseudonymous forums where users engage in batching. We show that standard notions of differential privacy (DP)~\cite{diff_privacy_defn} cannot be satisfied in our problem setting without introducing fake events or withholding events indefinitely. Therefore, we consider a ``one-sided'' relaxation of traditional DP \cite{osdp_2020}. Our formulation aims to prevent an adversary from inferring when batching happened, but allows an adversary to learn that batching did not happen. 
    \item We propose a general framework for designing mechanisms that guarantee one-sided DP by adding independent random delay to batched and unbatched comments. We show that we can instantiate this framework with a number of different distributions and guarantee privacy. Notably, it is possible to guarantee privacy with non-negative versions of typical distributions used for differential privacy like the Laplace distribution and the Staircase distribution. It is also possible to guarantee privacy at any setting of the privacy parameters by adding delay drawn from a uniform distribution with inflated probability mass at $0$, which we call the Zero-Inflated Uniform Mechanism.
    
    \item We establish the optimality of our Zero-Inflated Uniform Mechanism among mechanisms that add independent noise to each comment. In particular, we give a full characterization of the Pareto frontier of the expected delay added to batched and unbatched events by any mechanism that adds independent non-negative noise to comments, at any setting of privacy parameters, and show that our proposed mechanism achieves this frontier. This result may be of independent interest. While the uniform distribution is not typically used in the design of two-sided DP algorithms, our results show that for one-sided DP when only non-negative noise can be added (as is the case for streaming timing data) the Zero-Inflated Uniform Mechanism can optimally trade off privacy for utility.
    \item We conduct a series of experiments simulating linkage attacks using batched timing of tasks on Wikipedia article revision data and Bitcoin transaction data. These experiments reveal the applicability of our methods in preserving privacy in practice without exceedingly large delays. 
\end{itemize}
All of our code is available online at \href{https://github.com/akgoldberg/batching-privacy}{https://github.com/akgoldberg/batching-privacy}.

\section{Related Work}

There is a substantial body of work on anonymity when sending packets over a network. However, as we discuss below, the techniques developed therein are inapplicable to our setting. Specifically, prior work has described deanonymization attacks which leverage correlated timing of packet arrivals. The work gives various defenses against such attacks~\cite{shmatikov2006timing,levine2004timing, kadloor2013timing, javidbakht2017delayanonymity, stadium, vuvuzela}. Anonymous networking seeks to prevent an adversary from inferring the sender and recipient of a given message. Packets are routed through a sequence of ``mix nodes'' to obscure the path taken. The highly correlated arrival times of packets on the first mix node and the last mix node in one path can enable inferences that a specific sender and recipient are communicating with one another. Prior work~\cite{shmatikov2006timing,levine2004timing} demonstrates the practical viability of deanonymization attacks that take advantage of batching in anonymous networks. 

The defenses proposed in these papers rely on the introduction of \emph{dummy packets} or ``cover traffic'' to a network, obscuring any instance of batching amidst many instances of spurious batching. In contrast, a critical constraint in the settings we consider is the \textit{infeasibility of generating fake data} as a means of preserving privacy. Therefore, our work will consider mechanisms that delay batched arrivals in order to preserve anonymity, trading off delay for privacy, without introducing any synthetic data.

Our work defines privacy based on a ``one-sided'' relaxation of the popular notion of differential privacy~\cite{diff_privacy_defn}. The definition of one-sided DP was introduced in the paper~\cite{osdp_2020} in a setting where contributors of individual data-points to a database have different privacy constraints and hence data-points can be classified as ``sensitive'' and ``non-sensitive.'' In our work, we argue that this classification of sensitive and non-sensitive data-points is applicable to batched and unbatched events. Interestingly, while the paper \cite{osdp_2020} shows that one-sided DP can improve utility compared to standard two-sided DP, we find that in our problem setting, one-sided DP admits useful privacy-preserving algorithms where two-sided DP does not admit any useful algorithms at all. We cannot readily apply algorithms from the paper \cite{osdp_2020} due to the constraint that we publish all data. Therefore, while they develop mechanisms that release a subset of non-sensitive data with no noise addition, while withholding all sensitive data entirely, we consider mechanisms that add noise to both sensitive and non-sensitive data-points and release all data-points.

Geng and Viswanath~\cite{staircase} address the question of optimal distributions for noise addition in standard differential privacy. They show that in order to minimize the magnitude of noise added to a query with known sensitivity, noise should be drawn from a         ``staircase" distribution, which has a probability density function that is roughly a piece-wise constant approximation of the Laplace distribution. Our work can be seen as an analogous result in the one-sided DP regime. Specifically, we prove that for the one-sided relaxation of differential privacy, adding staircase noise is no longer optimal, but rather adding uniform noise with a possibly inflated probability of sampling $0$ minimizes the magnitude of noise addition.

Our running application in this paper is that of peer review. A few previous papers have considered certain issues of privacy in peer review, but with very different objectives and methods. The paper~\cite{ding2022calibration} considers the problem of miscalibration~\cite{flach2010kdd,roos2012statistical,ge13bias,wang2018your} in peer review. They consider privacy leakage when correcting for such miscalibration and provide methods (for a simplified setting) to mitigate this leakage. The paper~\cite{ding2020privacy} provides privacy-preserving algorithms for releasing some peer-review data to allow researchers at large to analyze and address problems like subjectivity~\cite{lee2015commensuration,noothigattu2018choosing} and miscalibration. The paper~\cite{jecmen2020manipulation} considers the problem of coalition-based fraud~\cite{Vijaykumar2020ASPLOS,littman2021collusion,wu2021making,jecmen2022tradeoffs} in peer review, and provides a randomized algorithm to assign reviewers to papers to mitigate such fraud. They argue that such a randomized assignment algorithm has another benefit: it can allow for release of the data that underlies the automated assignment algorithm while still preserving some privacy about which paper was assigned to which reviewer. We refer the reader to~\cite{shah2021survey} for an overview of research on peer review. 

\section{Problem Formulation}

We now describe our problem formulation. For clarity of exposition, we use the running example of peer-review.

\medskip{\bf Comment Arrivals.} We call the event when a reviewer makes a comment on a paper a \emph{comment arrival}. Each comment arrival consists of $4$ elements: the text of the comment, a timestamp $\ts$ when the comment arrived, a paper $\paper$ to which it responds, and the reviewer $\reviewer$ who made the comment. We assume that comments arrive in continuous time over an infinite time horizon, as this is the most general setup, although our analysis extends to any finite time horizon (for example, in the case where a conference has an end time after which comments can no longer be posted). We consider settings where the comments are publicly observable, as is the case for many conferences run on popular platforms like OpenReview.net. 

\medskip{\bf Batching.} In our initial model, we consider comments to be ``batched'' if they arrive simultaneously. Specifically, a set of $2$ or more  comment arrivals is \emph{batched} if all comments in the set come from the same reviewer at the same time, and furthermore, the comments are all on different papers. In Section~\ref{sec:extension_non_simultaenous} we discuss how to extend the model to allow for a short gap between batched comments. 

\medskip{\bf Comment Posting Mechanism.} A comment posting mechanism $\mech$ receives comments as they arrive and can choose to delay when they are posted, with the comments only becoming publicly visible at the time they are posted. The mechanism receives a streaming set of comment arrivals $\Arriv$ as input. It outputs a set of comments where each comment has identical content, paper, and reviewer to a comment in the input but with a potentially delayed timestamp. We place the following natural constraints on any \emph{valid} comment posting mechanism: \begin{enumerate}
    \item (\emph{Delay-Only}) If a comment arrives at time $\ts$ it must be output at time $\ts$ or later.
    \item (\emph{No Fake Data}) Any comment posted at time $\ts$ must have arrived at or before time $\ts$. 
    \item (\emph{Eventual Release of All Comments}) For any comment, letting $\delay$ denote the potentially randomized delay introduced to the comment by the mechanism, it must be that $\lim\limits_{\maxDelay \to \infty} \Pr[\delay \leq \maxDelay] = 1$.
\end{enumerate}

\medskip{\bf Privacy.} Our goal is to protect against an adversary who is trying to infer whether a specific pair of comment arrivals was batched. Following the widely-adopted framework of differential privacy, we consider a strong adversary who knows exactly when all comments arrived, except for one pair of comments that either arrived in a batch or at separate times. The adversary knows that comments arrive at the same time if batched and knows the exact inter-arrival time of the pair of comments if they arrive unbatched. In preserving privacy against such a strong adversary, we also provide privacy guarantees for general classes of weaker adversaries with less prior knowledge. For instance, in Section~\ref{sec:implementation_params} we discuss an adversary who only has an estimate of the baseline distribution of inter-arrival times when comments are unbatched, rather than the exact inter-arrival time.

Ideally, we would like to provide a privacy guarantee with respect to the standard notion of differential privacy (DP). Such a DP guarantee would promise difficulty of distinguishing whether the mechanism was run on one of two neighboring inputs, where one input has an additional batched pair of comments compared to its neighbor. Unfortunately, as we prove in Section~\ref{sec:dp_impossibility}, it is impossible to guarantee standard $\eps$-DP in this setting. There are two main reasons for this impossibility. First, consider defining neighboring inputs to a DP mechanism where a pair of comments arrives simultaneously in one input when batched, but arbitrarily far apart when unbatched in the neighboring input. Then, to satisfy a traditional DP guarantee, batched comments must be delayed \emph{indefinitely} to make these two inputs indistinguishable. Second, even with a bounded change in arrival time for any comment on neighboring inputs, we show that if the neighboring relation is symmetric (i.e., a pair of comments can be batched in one input and unbatched in the other, and it doesn't matter which input contains the batched comments), then to satisfy $\eps$-DP the mechanism must delay a batched comment indefinitely.

In order to address the aforementioned roadblocks, we relax the definition of neighboring inputs in two ways. First, we introduce a real-valued parameter $\gap >0$ into our formulation of neighbors that bounds how far in time a batched comment can move in a neighboring input where it arrives unbatched. Second, we define neighbors in a one-sided manner: a set of comment arrivals neighbors another set only if it contains one \emph{additional} pair of batched comments as compared to its neighbor. In contrast, a set of comment arrivals does not neighbor another set if it contains one \emph{fewer} pair of batched comments than its potential neighbor. Formally, we define neighboring comment arrival sets as follows:

\begin{definition}[$\gap$-Neighboring Comment Arrival Sets]
A set of comment arrivals $\inpBatch$ is \emph{$\gap$-neighboring} to set of comment arrivals $\inp$, if $\inpBatch$ can be obtained from $\inp$ by batching together one pair of comments that arrive separately in $\inp$. The comments must arrive within $\gap$ units of time of one another in $\inp$ and the later comment moves to the earlier comment in $\inpBatch$ to create a batch. Formally, $\exists (\act, \ts, \paper, \reviewer), (\act', \ts', \paper', \reviewer) \in \inp$ such that $\paper \not = \paper'$, $0 < \ts' - \ts \leq \gap$ and $\inpBatch = (\inp \setminus \{c'\}) \cup \{(c', \ts, \paper', \reviewer)\}$. 
\end{definition}

Note that this definition of adjacency is asymmetric as a set of comment arrivals with no pairs of batched comments is not $\gap$-adjacent to any other sets of comment arrivals. As an example, consider the following pair of comment arrival sets $\inp$ and $\inpBatch$:
\begin{align*}
    \inp = &  \{(\act_1, \ts = 1, \paper_1, \reviewer_1), (\act_2, \ts = 2, \paper_1, \reviewer_2), \pmb{(\act_3, \ts = 3, \paper_2, \reviewer_1)}\},\quad\text{and} \\ 
    \inpBatch = & \{(\act_1, \ts = 1, \paper_1, \reviewer_1), \pmb{(\act_3, \ts = 1, \paper_2, \reviewer_1)}, (\act_2, \ts=2, \paper_1, \reviewer_2)\}.
\end{align*}
Then under our definition above, $\inpBatch$ is $2$-neighboring to $\inp$. However, $\inp$ is \emph{not} $2$-neighboring to $\inpBatch$.\footnote{The reader may have observed that the definition of neighboring comment arrival sets has a technical condition that a batched comment moves later in time in a neighboring input with one fewer instance of batching. It is possible to modify the formulation to let a batched pair of comments arrive at either one of the later or earlier arrival times of an unbatched pair in an adjacent input. This modified formulation would capture an even stronger adversary who knows the exact time-frame in which a batched pair arrives. However, ensuring privacy against this adversary would require even more delay added to the system, Hence, we do not pursue this formulation.}

Now, we define  privacy of a mechanism using a notion similar to the definition of one-sided differential privacy introduced in~\cite{osdp_2020}. We note that apart from the one-sidedness of neighbors, our privacy formulation differs substantially from that of~\cite{osdp_2020} as we focus on inputs differing in the timing of a pair of comments due to batching, while~\cite{osdp_2020} considers databases where arbitrary entries are considered non-private.

For any finite time horizon $\finiteTime$ and set of comment arrivals $\Arriv$, we will let $\mech_\finiteTime(\Arriv)$ denote the output of the mechanism up to time $\finiteTime$. Then, we define privacy as follows:

\begin{definition}[$(\eps, \gap)$-One-Sided Differential Privacy (OSDP)]
\label{defn:osdp}
For any $\eps \geq 0$ and $\gap > 0$, a comment posting mechanism $\mech$ is $(\eps, \gap)$-one-sided differentially private if for any $\inp, \inpBatch$ such that $\inpBatch$ is $\gap$-neighboring to $\inp$, for any time horizon $\finiteTime$, and for any subset of possible outputs $\out \subseteq \text{Range}(\mech_\finiteTime)$ of the mechanism: \begin{align*}\Pr[\mech_\finiteTime(\inpBatch) \in \out] \leq e^{\eps} \Pr[\mech_\finiteTime(\inp) \in \out].
\end{align*}
\end{definition}

This privacy definition guarantees that the likelihood of observing an outcome on an input with at least one instance of batching is never much larger than the likelihood of observing that outcome on an input with one \emph{fewer} batched pair. Therefore, the mechanism obscures the fact that any pair of comments was batched. However, it is possible for the mechanism to reveal that a pair of comments was unbatched; we allow for outputs that occur with non-zero probability given input $\inp$ but zero probability given input $\inpBatch$ (unlike in standard two-sided DP). We argue that the one-sided definition effectively captures privacy risk due to batching, as the presence of a batched pair of comments is sensitive information, while the absence of batching is non-sensitive. We further discuss the motivation for only treating batching as sensitive via the concrete example of reviewer deanonymization by a meta-reviewer. 

The privacy definition requires two parameters: $\eps$ and $\gap$. The interpretation of $\eps$ is similar to two-sided DP as it quantifies the ``level'' of privacy: for smaller $\eps$ it is harder to distinguish neighboring inputs, whereas for larger $\eps$ it is easier to distinguish neighboring inputs. The $\gap$ parameter captures domain knowledge about what types of inputs can be neighbors, similar to restricting the domain of inputs in two-sided DP. Roughly, $\gap$ should capture how far apart consecutive comments would plausibly arrive if batching were not occurring. It is necessary for a practitioner to include this domain knowledge in the form of finite value $\gap$ as we prove that batched comments must be delayed by at least $\gap$ (in Section~\ref{sec:pareto}) and hence without this bound, comments must be withheld indefinitely. We give heuristics for how to set $\gap$ based on a hypothesis testing interpretation of the privacy definition in Section~\ref{sec:implementation_params}.

~\\\textbf{Utility.} We measure the cost of our mechanism in terms of expected delay added to comments. Because the privacy guarantee is asymmetric, the mechanism can behave differently on batched and unbatched comments. Therefore, we will consider measuring utility in terms of expected delay to batched comments denoted $\E[\distrB]$, expected delay to unbatched comments denoted $\E[\distrU]$ or more generally any weighted sum of the two expectations. 

~\\\textbf{Goal.} Our goal is to design comment posting mechanisms that guarantee $(\eps, \gap)$-one-sided differential privacy for chosen privacy parameters $\eps$ and $\gap$ while minimizing the expected delay added to comments. We may add random delay to batched and unbatched comments drawn from different distributions $\distrB$ and $\distrU$  respectively. Therefore, we wish to design $(\eps, \gap)$-OSDP mechanisms that are Pareto optimal in trading off between $\E[\distrB]$ and $\E[\distrU]$ at any setting of $\eps$ and $\gap$. Moreover, we want to allow practitioners to choose a mechanism on this Pareto frontier that minimizes an appropriate cost function suiting the requirements of their system. For instance, a system with a higher rate of batching may wish to weight delay to batched comments higher in their cost function than a system with a lower rate of batching. To this end, we consider minimizing any cost function that is a convex combination of expected delay to batched and unbatched comments. We aim to provide the exact mechanism on the Pareto frontier that minimizes $\weightB \E[\distrB] + (1-\weightB)\E[\distrU]$ for any choice of weighting parameter $\weightB \in [0,1]$ and any privacy parameters $\eps$ and $\gap$. We note that this choice of utility function is without loss of generality. In particular, the feasible region of $\E[\distrB]$ and $\E[\distrU]$ is convex (as we prove in Appendix~\ref{appendix:optimal_proof}, Lemma~\ref{lem:frontier_weighted_sum}.) Therefore, any mechanism that is Pareto optimal in trading off $\E[\distrB]$ and $\E[\distrU]$ minimizes the weighted cost function for some choice of $\weightB$ (since any point on the Pareto frontier of a convex feasible region optimizes some weighted sum objective per Boyd~\cite[Chapter 4.7]{boyd2004convex}).

\paragraph{Example: De-anonymizing reviewers.}
\label{sec:one_sided_privacy_risk}

\begin{figure}
    \centering
    \includegraphics[width=0.6\linewidth]{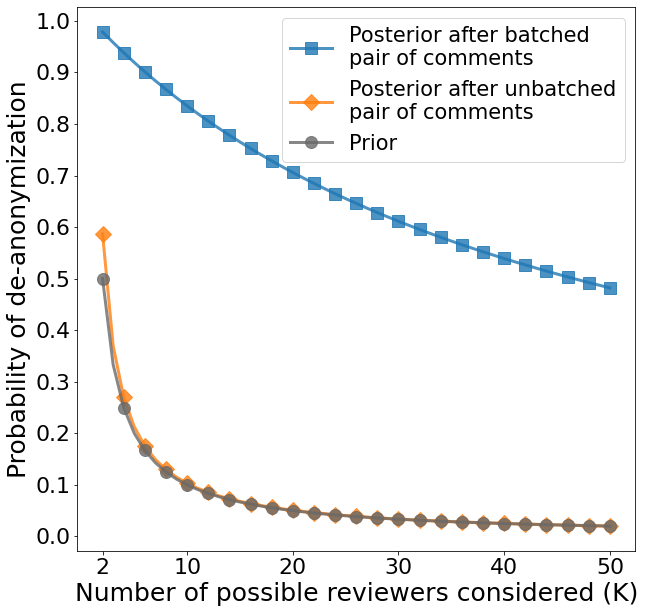}
    \caption{Success probability of de-anonymizing a (meta)-reviewer after learning that a pair of comments arrived together vs. learning that a pair of comments did not arrive together.}
    \label{fig:posteriors}
\end{figure}

We now discuss the one-sided nature of privacy risk inherent to batching using the running example of a meta-reviewer de-anonymizing a reviewer or meta-reviewer of a paper they have authored.  Recall the introductory scenario where an meta-reviewer observes two comments $\act$ and $\act'$ that arrive consecutively on different papers and are made by (meta)-reviewers $\reviewer_1$ and $\reviewer'$ respectively (where it is possible that $\reviewer' = \reviewer_1$). The meta-reviewer knows that the first comment was made by $\reviewer_1$ and has a uniform prior over $K$ possible reviewers who could have made $\act'$ (including $\reviewer_1$). They wish to de-anonymize $\reviewer'$ based on whether or not $\act'$ arrived in a batch with $\act$. From our aforementioned analysis of a conference peer review where we define two comments as ``arriving together'' if they arrive within $5$ minutes of one another, we estimate that: $\Pr[\act, \act' \text{ arrive together} \;|\; \reviewer' = \reviewer_1] \approx 0.3$, while $\Pr[\act, \act' \text{ arrive together} \;|\; \reviewer' \not= \reviewer_1] \approx 0.0066$. Therefore, after learning that $\act$ and $\act'$ arrived together, the meta-reviewer's posterior puts the most weight on $\Pr[\reviewer' = \reviewer_1 \;|\; \act, \act' \text{ arrived together}] = \frac{0.301}{0.301 + 0.0066(K-1)}$. On the other hand, after learning that $\act$ and $\act'$ did not arrive together, their posterior puts the most weight on: $\Pr[\reviewer' = \reviewer_k \;|\; \act_1, \act_2 \text{ did not arrive together}] = \frac{0.9934}{0.699 + 0.9934(K-1)} \; \text{for } K \not = 1$. We give further detail on how these statistics were estimated in Appendix~\ref{app:deanon_attack}.

As shown in Figure~\ref{fig:posteriors}, in learning that the pair of comments was batched, the meta-reviewer can identify the (meta)-reviewer of a paper they authored with much higher confidence than before observing the batched timing; on the other hand, by learning that the pair of comments was unbatched, the meta-reviewer's posterior hardly changes from the prior. 
Our one-sided privacy definition captures this asymmetric privacy risk. It ensures that an adversary does not learn much about the sensitive information of whether two comments are likely to be batched after observing the time that comments get posted, while allowing the adversary to potentially learn the insensitive information that two comments were unbatched.

\section{Theoretical Results}

In this section, we present our main theoretical results. First, in Section~\ref{sec:alg_framework}, we propose an algorithmic framework to design comment posting mechanisms that guarantee $(\eps, \gap)$-one-sided differential privacy under batching. In this framework, we add random noise to the timestamps of batched and unbatched comments, drawing the noise from a pair of distributions $(\distrB, \distrU)$ that depend on parameters $\eps$ and $\gap$.

Within this framework, there are many possible choices of the noise distributions $(\distrB, \distrU)$, and we investigate them in Section~\ref{sec:distributions}. For instance, one could use one-sided analogues of distributions commonly used for two-sided differential privacy, like exponential noise, which is the absolute value of the Laplace distribution~\cite{diff_privacy_defn}, or one-sided staircase noise~\cite{staircase} (whose two-sided version is known to be optimal for two-sided DP~\cite{staircase}). However, we show that perhaps surprisingly, these distributions are all sub-optimal for the privacy-delay trade-off. 

In Section~\ref{sec:pareto} we provide another distribution -- a zero-inflated uniform distribution with carefully chosen parameters -- that we show guarantees one-sided differential privacy in our setting and also achieves a Pareto-optimal privacy-delay trade-off. 

Finally, in Section~\ref{sec:dp_impossibility}, we motivate the usefulness of our one-sided DP formulation as a means of capturing the privacy-delay trade-off by showing that the popular two-sided definition of DP does not yield a useful privacy-delay trade-off for valid comment posting mechanisms. 

\subsection{Algorithmic Framework}
\label{sec:alg_framework}

In Algorithm~\ref{alg:dp_random_framework}, we present a general recipe for designing randomized delay mechanisms. The meta-algorithm receives as input privacy parameters $\eps$ and $\gap$ as well as probability distributions $\distrB$ and $\distrU$ that depend on $\eps$ and $\gap$. We will then prove that if pairs of distributions satisfy an ``indistinguishability'' property then Algorithm~\ref{alg:dp_random_framework} yields a $(\eps, \gap)$-OSDP mechanism. 

\begin{algorithm}
\caption{Framework for Designing a Randomized Delay Mechanism}
\label{alg:dp_random_framework}
\begin{algorithmic}
   \STATE {\bfseries Input:} privacy parameter $\eps > 0$, maximum time gap $\gap > 0$, noise addition distributions $\distrB$ and $\distrU$
   \FOR{each comment arrival time $\ts$}
    \IF{a set of batched comments arrives}
    \STATE For each comment, independently sample $\delay \sim \distrB(\eps/2, \gap)$ and post the action at time $\ts + \delay$. 
    \ELSIF{if an unbatched comment arrives}
    \STATE Post the comment at time $\ts + \delay$ where $\delay \sim \distrU(\eps/2, \gap)$
    \ENDIF 
    \ENDFOR 
\end{algorithmic}
\end{algorithm}

Mechanisms within this framework satisfy two useful qualitative properties for deployment in real applications. First, because the noise is sampled at arrival time, we can tell each user the duration of the delay on their comment as soon as they create it. Second, because the noise is sampled independently for each comment, the algorithm does not require a centralized coordinator to determine post times. This enables privacy-sensitive individuals to implement the algorithm for themselves. This ability to be implemented locally is a necessary property for use in cryptocurrencies where there is no central trusted server. 

Now, any choice of $(\distrB, \distrU)$ can satisfy $(\eps, \gap)$-OSDP as long as $\distrB$ and $\distrU$ are indistinguishable in the following sense:
\begin{definition}[One-Sided Indistinguishable Distributions]
\label{defn:indisting_distr}
Let $\distrB$ and $\distrU$ be non-negative random variables. We say that the ordered pair  $(\distrB, \distrU)$ is \emph{$(\eps, \gap)$-one-sided indistinguishable} if, for any measurable set $S \subseteq \R$ and any $\ts_0 \in [0, \gap]$, the distributions satisfy:
\begin{align*}
\Pr[\distrB \in S] \leq e^\eps \Pr[\distrU \in S-\ts_0],
\end{align*}
where for any $S \subseteq \R, t \in \R :\: S - t = \{s - t | s \in S\}$.
\end{definition}

The following theorem shows sufficiency of such one-sided indistinguishable distributions for guaranteeing privacy. 
\begin{theorem}[Privacy of Randomized Delay Mechanisms]
\label{thm:priv_alg}
Let $(\distrB, \distrU)$ be any pair of $(\eps/2, \gap)$-one-sided indistinguishable distributions. Then, Algorithm~\ref{alg:dp_random_framework} using $\distrB$  and $\distrU$ as noise-addition distributions guarantees $(\eps, \gap)$-one-sided differential privacy.
\end{theorem}

We give the proof of the above theorem in Appendix~\ref{app:priv_alg}. The proof follows by observing that in neighboring inputs, a pair of comments that was batched becomes unbatched with one comment arrival moved forward by at most $\gap$ time units. Hence, if $\distrB$ and $\distrU$ have a likelihood ratio bounded by $e^{\eps/2}$ for any values within $\gap$ time units of one another, it is hard to distinguish whether the mechanism was given an input with two unbatched comments arriving $\gap$ time units apart or two batched comments arriving at the same time (up to a multiplicative factor of $e^{\eps})$.

\subsection{Privacy-preserving delay distributions}
\label{sec:distributions}

We now describe a number of possible choices for $(\eps, \gap)$-one-sided indistinguishable distributions $(\distrB, \distrU)$ that can be used in our algorithmic framework. We show that we can use an exponential distribution, which is the one-sided version of the Laplace distribution. We can also add noise from the absolute value of the staircase distribution, which was proven in~\cite{staircase} to be optimal for noise addition in two-sided DP, giving smaller delay than the exponential. Alternatively, we can add noise to unbatched comments drawn from a zero-inflated uniform distribution where we add $0$ delay with probability $1-\pUnif$ (for some parameter $\pUnif$) and delay drawn from a uniform distribution with probability $\pUnif$.

\begin{theorem}[Choices of One-Sided Indistinguishable Distributions]
\label{thm:privacy_distributions}
The following choices of $\distrB$ and $\distrU$ are $(\eps, \gap)$-one-sided indistinguishable:
\begin{enumerate}[label=(\arabic*)]
    \item \emph{Exponential}\footnote{In the notation to follow, we parameterize the exponential distribution by its rate.}: $\distrB = \gap + \text{Exponential}(\eps / \gap)$, $\distrU = \text{Exponential}(\eps/\gap)$ 
    \item \emph{Staircase~\cite{staircase}}\footnote{The staircase distribution is parameterized by $3$ values $\eps, \Delta, \gap$ in \cite{staircase}. Here, we take $\text{Staircase}(\eps, \gap)$ to mean the staircase distribution with $\eps = \eps$, $\Delta = \gap$ and $\gamma = \frac{1}{1+e^{\eps/2}}$, which is the optimal value of $\gamma$ to minimize expectation per \cite{staircase}.}: $\distrB = \gap + |\text{Staircase}(\eps, \gap)|$, $\distrU = |\text{Staircase}(\eps, \gap)|$ 
    \item \emph{Uniform}: $\distrB = \Unif(\gap, \frac{1}{1 - e^{-\eps}}  \gap)$, $\distrU = \Unif(0,  \frac{1}{1 - e^{-\eps}}  \gap)$
    \item \emph{Zero-inflated Uniform with parameter $\pUnif$.} For $e^{-\eps} <  \pUnif \leq 1$: \begin{align*}
    \distrB &= \text{Uniform}\left(\gap, \tfrac{\pUnif}{\pUnif-e^{-\eps}}  \gap\right) \\
    \distrU &= \begin{cases}
    0 & \text{with probability } 1-\pUnif \\ 
    \text{Uniform}\left(0, \tfrac{\pUnif}{\pUnif-e^{-\eps}}  \gap\right) & \text{with probability } \pUnif.
    \end{cases}
\end{align*}
\end{enumerate}
These choices of $(\distrB, \distrU)$ incur the following expected delays:
\begin{enumerate}[label=(\arabic*)]
    \item \emph{Exponential:} $\E[\distrB] = \gap (1 + \frac{1}{\eps})$ and $\E[\distrU] = \gap  \frac{1}{\eps}$
      \item \emph{Staircase}: $\E[\distrB] = \gap (1 + \frac{e^{\eps/2}}{e^\eps - 1})$, $\E[\distrU] = \gap  \frac{e^{\eps/2}}{e^\eps - 1}$
    \item \emph{Uniform:} $\E[\distrB] = \frac{1}{2}\gap \left(1 + \frac{e^\eps}{e^\eps - 1} \right)$  and $\E[\distrU] = \frac{1}{2}\gap \left(\frac{e^\eps}{e^\eps - 1}\right)$
    \item \emph{Zero-inflated Uniform with parameter $\pUnif$}: $\E[\distrB] = \frac{1}{2} \gap \left(\pUnif + \frac{\pUnif e^\eps}{\pUnif e^\eps -1} \right)$ and $\E[\distrU] = \frac{1}{2} \gap \left( \frac{\pUnif^2 e^\eps}{\pUnif e^\eps - 1} \right)$.
\end{enumerate}
\end{theorem}

The proof of the above theorem can be found in Appendix~\ref{app:pf_indistinguish}. Note that the (uniform, uniform) noise additions are a special case of (uniform, zero-inflated uniform) taking $\pUnif = 1$. We highlight them separately in Section~\ref{thm:privacy_distributions} as we introduce the zero-inflated uniform distribution for the first time here. In the next section, we show that a zero-inflated uniform distribution is Pareto optimal for appropriate choice of $\pUnif$.

Notably, the choice of parameters for the exponential and staircase distributions given in Theorem~\ref{thm:privacy_distributions} are the optimal choice of parameters in the sense that they minimize expected delay at fixed values of privacy parameters $\eps$ and $\gap$ when adding i.i.d. exponential or staircase noise plus a constant offset to all comments:

\begin{theorem}[Optimal Choice of Parameters for the Exponential and Staircase Distributions]

\label{thm:distr_params}
Let $\distrB, \distrU$ be non-negative noise-addition distributions that guarantee $(\eps, \gap)$-OSDP when used in Algorithm~\ref{alg:dp_random_framework} where $\distrB = \distrOffset_\distrB + \distrSingle$ and $\distrU = \distrOffset_\distrU + \distrSingle$ for constants $\distrOffset_\distrB, \distrOffset_\distrU > 0$ and non-negative random variable $\distrSingle$. Then, if $\distrSingle$ is an exponential random variable or a staircase random variable, $\E[\distrB]$ and $\E[\distrU]$ are minimized at any values of $\eps, \gap$ by the choice of parameters in Theorem~\ref{thm:privacy_distributions} such that $\distrB, \distrU$ are $(\eps/2, \gap)$-one-sided indistinguishable.  
\end{theorem}

The proof of the above theorem can be found in Appendix~\ref{appendix:exp_staircase_proof}. By Theorem~\ref{thm:privacy_distributions} and Theorem~\ref{thm:distr_params}, adding i.i.d. exponential or staircase noise plus a constant offset is strictly sub-optimal in minimizing expected delay as zero-inflated uniform noise can achieve lower delay at the same privacy level.

\begin{corollary}
Among $(\eps, \gap)$-OSDP mechanisms following the framework of Algorithm~\ref{alg:dp_random_framework}, taking $\distrB$ and $\distrU$ to be i.i.d. exponential or staircase distributions (with constant offsets) is strictly sub-optimal in minimizing $\E[\distrB]$ and $\E[\distrU]$ for any values of $\eps$ and $\gap$. In particular, using the zero-inflated uniform mechanism with appropriate choice of $\pUnif$ can achieve lower expected delay for both $\E[\distrB]$ and $\E[\distrU]$ at any values of privacy parameters $\eps$ and $\gap$.
\end{corollary}

In this setting, the exponential and staircase distributions typically used in two-sided DP add significantly more delay than zero-inflated uniform noise, especially at small values of $\eps$. In Figure~\ref{fig:exp_delay}, we show the expected delay for the optimal exponential, staircase, uniform, and zero-inflated uniform at each setting of $\eps$. For both batched and unbatched comments, the uniform and zero-inflated uniform distributions add a factor of nearly two times less delay than the staircase and exponential at small values of $\eps$. For larger values of $\eps$, all of the aforementioned distributions add similar delay, with uniform adding the least delay to batched comments and the zero-inflated uniform adding the least delay to unbatched. In the next section, we formally prove that zero-inflated uniform noise is Pareto optimal and characterize the optimal choice of $\pUnif$ for any objective function that is a weighted sum of $\E[\distrU]$ and $\E[\distrB]$ based on the setting of $\eps$.

\begin{figure}
\centering
\begin{subfigure}{.5\textwidth}
  \centering
  \includegraphics[width=.8\linewidth]{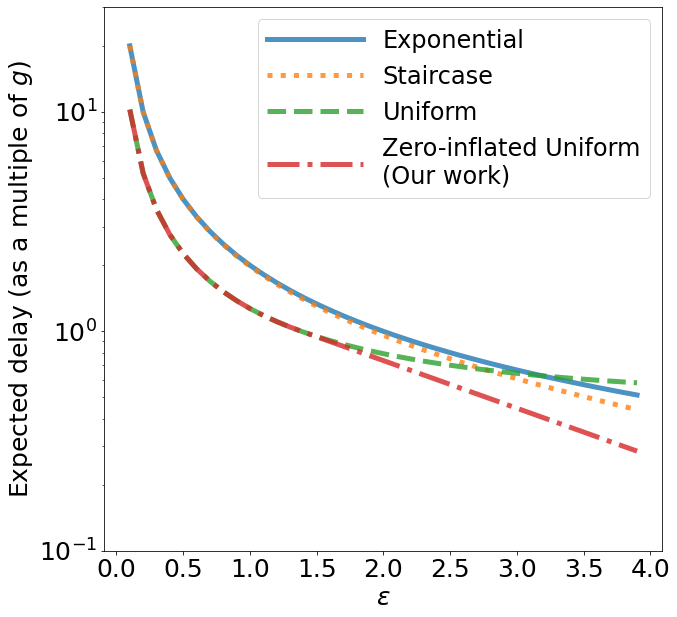}
  \caption{Expected delay to unbatched comments}
  \label{fig:unbatched_expected_delay}
\end{subfigure}%
\begin{subfigure}{.5\textwidth}
  \centering
  \includegraphics[width=.8\linewidth]{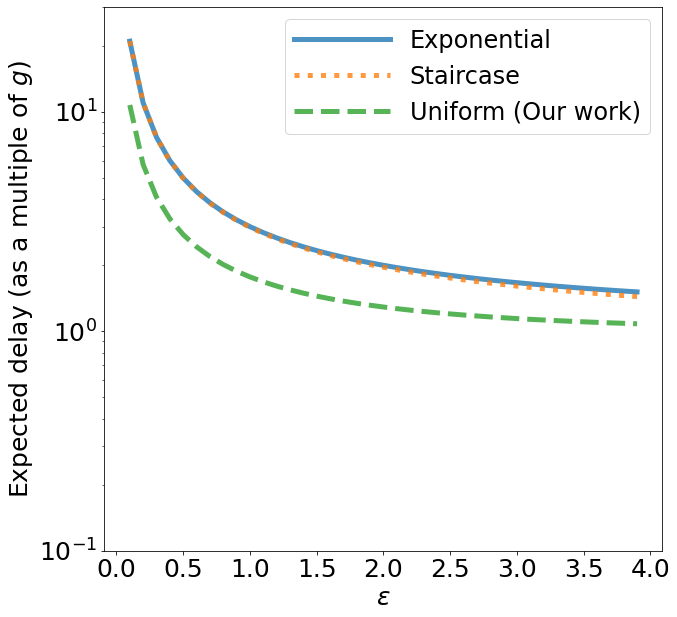}
  \caption{Expected delay to batched comments}
  \label{fig:batched_expected_delay}
\end{subfigure}
\caption{Expected delay (as a multiple of $\gap$) for Algorithm~\ref{alg:dp_random_framework} with the exponential, staircase, zero-inflated uniform and uniform distributions at varying values of $\eps$. All distributions use the optimal setting of parameters at a given $\eps$. The zero-inflated uniform parameter is chosen to minimize delay to unbatched comments. The delay (y axis) is plotted on a log scale.}
\label{fig:exp_delay}
\end{figure}

\subsection{Pareto-optimal Algorithm}
\label{sec:pareto}

\begin{algorithm}[tb]
\caption{Zero-Inflated Uniform Mechanism}
\label{alg:unif_mechanism}
\begin{algorithmic}
   \STATE {\bfseries Input:} privacy parameter $\eps > 0$, maximum time gap $\gap > 0$, weighting of expected delay to batched comments $\weightB \in [0,1]$
   \STATE 
   \STATE Set $\pUnif = \min\left\{e^{-\eps/2} \left(1 + \sqrt{1 + e^{\eps/2} \frac{\weightB}{1 - \weightB}} \right), 1 \right\}$
    \STATE
   \FOR{each comment arrival time $\ts$}
    \IF{a set of batched comments arrives}
    \STATE For each comment, independently sample $\delay \sim \Unif\left(\gap, \tfrac{\pUnif}{\pUnif-e^{-\eps/2}}  \gap\right)$ and post the comment at time $\ts + \delay$. 
    \ELSIF{an unbatched comment arrives}
    \STATE Post the comment at time $\ts + \delay$ where $\delay = 0$ with probability $1-\pUnif$ and $\delay \sim \Unif\left(0, \tfrac{\pUnif}{\pUnif-e^{-\eps/2}}  \gap\right)$ with probability $\pUnif$.
    \ENDIF 
    \ENDFOR 
\end{algorithmic}
\end{algorithm}

In this section, we derive the Pareto frontier (trading off the expected delay for batched and unbatched comments) of noise-addition distributions for a given $(\eps, \gap)$-one-sided indistinguishability constraint (Definition~\ref{defn:indisting_distr}). We show that adding zero-inflated uniform noise with appropriate choice of parameter $\pUnif$ achieves optimal expected delay among mechanisms that add independent noise to each comment. While the optimality result holds only within the class of mechanisms that adds \emph{independent} noise to each comment, this constraint allows for an algorithm to be implemented locally without requiring coordination by a centralized server. This constraint is a common property of many deployed privacy-preserving algorithms. For instance, local differential privacy \cite{local-dp} requires that randomization needed for privacy is added locally by each holder of a data-point, and the Tor anonymous network \cite{tor} protocol requires that initiators of connections choose the (random) path on which to send a message themselves.

Given an $(\eps, \gap)$-one-sided privacy constraint, our algorithmic framework (Algorithm~\ref{alg:dp_random_framework}) has many choices of noise-addition distributions that can guarantee privacy. In terms of delay, there are two quantities to optimize -- the delay incurred by batched comments and that incurred by unbatched comments. A natural utility objective to consider is a convex combination of the two expectations: 
\begin{align*}
\weightB  \E[\distrB] + (1-\weightB)  \E[\distrU],\qquad\text{for a given parameter }\weightB \in [0,1].
\end{align*}
The parameter $\weightB \in [0,1]$ determines how much weight is given to batched comments in the utility function. For example, a user of our algorithm may estimate the relative rate of batching in the system and set $\weightB$ to this value to optimize for the overall average expected delay across all comments. 

We present our main algorithm as Algorithm~\ref{alg:unif_mechanism}. Our algorithm follows our previously introduced framework (Algorithm~\ref{alg:dp_random_framework}). It chooses $\distrU$ as a zero-inflated uniform distribution with a carefully chosen value of parameter $\pUnif$ (dependent on $\eps$ and $\weightB$), and chooses $\distrB$ as a uniform distribution. The following theorem now  proves that for any privacy parameters our algorithm is indeed Pareto optimal -- it optimally trades off privacy and unbatched delay and batched delay.
\begin{theorem}[Pareto optimality of the Zero-Inflated Uniform Mechanism]
\label{thm:unif_opt}
Algorithm~\ref{alg:unif_mechanism} is Pareto optimal between expected delay to batched and unbatched comments at a given setting of $(\eps, \gap)$ among valid $(\eps, \gap)$-OSDP mechanisms that add independent noise to each comment. Further, given weight parameter $\weightB \in [0,1]$ and privacy parameters $(\eps, \gap)$ as input, Algorithm~\ref{alg:unif_mechanism} minimizes cost function $\weightB  \E[\distrB] + (1- \weightB)  \E[\distrU]$ at any given privacy level $(\eps, \gap)$ among mechanisms adding independent noise drawn from distributions $\distrB$ and $\distrU$ to batched and unbatched comments respectively.
\end{theorem}

\noindent We give a proof sketch below, for the full proof see Appendix~\ref{appendix:optimal_proof}.
\begin{sproof}
Roughly, the proof proceeds as follows:
\begin{itemize}
    \item We consider any $(\eps, \gap)$-indistinguishable noise addition distributions $(\distrB, \distrU)$ added to batched and unbatched comments respectively. Using results from \cite{staircase}, we argue that for large enough $i\in \N$,  we can approximate $\distrB$ and $\distrU$ arbitrarily well with random variables that have piece-wise constant probability density functions and each constant interval has length $\gap/i$. 
    \item We establish properties of any Pareto optimal $(\distrB_i, \distrU_i)$ by directly proving that we can decrease the expectation of both $\E[\distrB_i]$ and $\E[\distrU_i]$ for any pair of distributions that violates these properties. Taken together the properties yield the exact form of any Pareto optimal $\distrB_i$ and $\distrU_i$. Taking limits as $i \to \infty$ gives that the Pareto frontier is realized by uniform and zero-inflated uniform distributions for some setting of $\eta$.
The proof follows by directly proving that we can decrease the expectation of both $\E[\distrB_i]$ and $\E[\distrU_i]$ for any pair of distributions that violates these properties.
\item Finally, we analytically solve for the value of parameter $\eta$ in the zero-inflated uniform distribution that minimizes weighted objective $\weightB \E[\distrB] + (1-\weightB) \E[\distrU]$ for any $\weightB \in [0,1]$. 
\end{itemize}
\end{sproof}

As shown in Figure~\ref{fig:pareto_frontier}, for smaller privacy budgets where $\eps \leq 2\ln(2)$, there is a single point on the Pareto frontier. Adding uniform noise with no inflated probability mass at $0$ minimizes $\E[\distrB]$ and $\E[\distrU]$ simultaneously. For larger $\eps$, it is possible to trade off between $\E[\distrB]$ and $\E[\distrU]$, achieving near-zero delay to unbatched comments. In practice, a user can decide what value of $\pUnif$ to use based on their preferred convex combination of $\E[B]$ and $\E[U]$. 

Note that our result holds for all mechanisms that add independent noise to each comment, as the delay added to comments must be $(\eps/2, \gap)$-one-sided indistinguishable to preserve privacy. Therefore, the zero-inflated uniform mechanism (Algorithm~\ref{alg:unif_mechanism}) is the Pareto optimal mechanism among this class of algorithms. It may be possible to add even less delay with mechanisms that can coordinate across comments and correlate noise addition. We leave this question open for future work.

\begin{figure}
    \centering
    \includegraphics[width=0.7\linewidth]{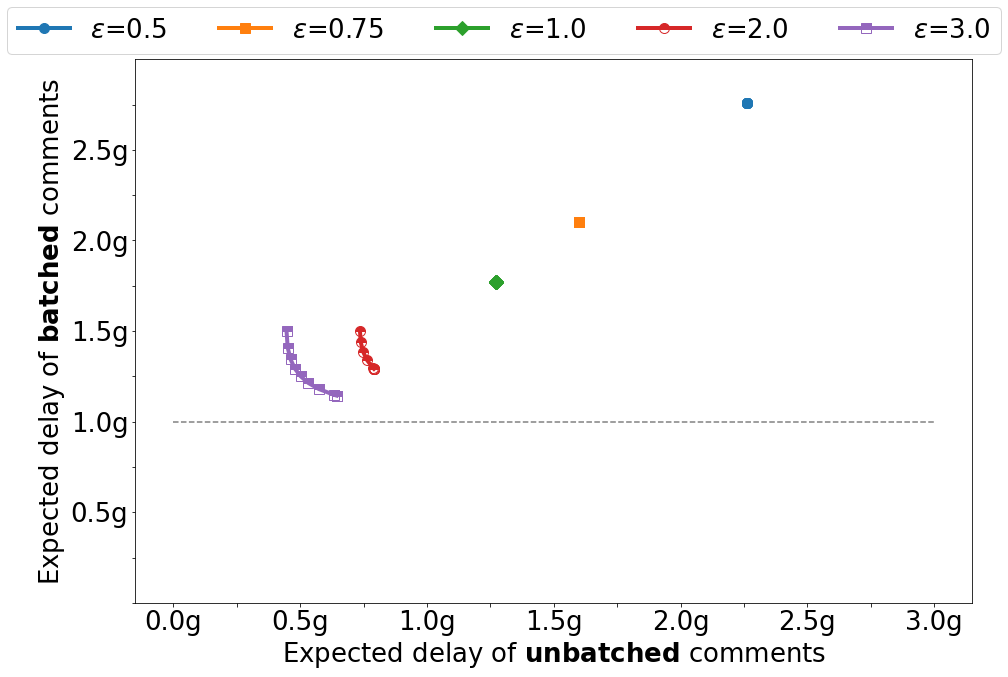}
    \caption{Pareto frontier for the expected delay added to batched and unbatched comments $(\E[\distrB]$ and $\E[\distrU]$) at different values of privacy parameter $\eps$.}
    \label{fig:pareto_frontier}
\end{figure}

\subsection{Impossibility of ``Two-Sided'' Differential Privacy}
\label{sec:dp_impossibility}

In the prior sections, we have characterized the privacy-utility trade-off for the one-sided relaxation of differential privacy. One might wish to obtain similar results for the standard two-sided definition of differential privacy, which would provide even stronger privacy guarantees. In this section, we prove the impossibility of guaranteeing two-sided differential privacy under the constraints of a valid comment postinging mechanism. These results motivate the further modeling assumptions on the adversary's prior knowledge about batching and attempted attacks that are used in the definition of $(\eps, \gap)$-OSDP.

First, we recall the standard definition of two-sided differential privacy. The key difference between this definition and our one-sided Definition \ref{defn:osdp} is in the formulation of ``neighboring'' inputs. In our one-sided definition, we use an asymmetric relation for neighboring inputs where one input with an additional batched pair of comments neighbors an input with one fewer pair. This captures the notion that batching is sensitive while the absence of batching is insensitive. For two-sided DP, we will give a definition with an abstract notion of neighbors and then concretely instantiate this definition with different possible notions of neighboring inputs. Critically, we will consider \emph{symmetric} relations for neighboring inputs in the definition of two-sided DP. This corresponds to preventing an adversary from inferring both whether batching occurred and whether batching did not occur.

Recall that $\mech_\finiteTime(\Arriv)$ denotes the output of the mechanism up to time $\finiteTime$. Then:

\begin{definition}[Two-Sided Differential Privacy for Batched Arrivals:]
For any $\eps \geq 0, \delta \in [0,1]$, a comment posting mechanism $\mech$ is $(\eps, \delta)$-differentially private if, for any time horizon $\finiteTime$ and for any subset  $\out \subseteq \text{Range}(\mech_\finiteTime)$ of possible outputs of the mechanism: \begin{align*}\Pr[\mech_\finiteTime(\inp') \in \out] \leq e^{\eps} \Pr[\mech_\finiteTime(\inp) \in \out],\end{align*} where $\inp$ and $\inp'$ are two \emph{``neighboring''} sets of comment arrivals.
\end{definition}

Now, we state our main impossibility result. We consider three natural definitions of neighboring sets of comment arrivals. The first definition adds or removes a comment from the set of comment arrivals corresponding to the notion of ``unbounded'' differential privacy in the literature \cite{dp_defn_unbounded}. The second definition moves a comment from being batched to unbatched in neighboring inputs by changing its timestamp, corresponding  to the notion of ``bounded'' differential privacy in the literature \cite{diff_privacy_defn}. Finally, the third definition restricts the second definition of neighbors further by placing a bound on how far a comment can move (which we call $\gap$), similar to the practice of constraining the domain of possible inputs to a differentially private mechanism. We show that it is not possible to guarantee privacy for any of these notions of neighbors:

\begin{theorem}[Impossibility of Two-Sided Differential Privacy]
\label{thm:impossibility}
For any of the following natural definitions of ``neighboring'' sets of comment arrivals, there is no two-sided differentially private, valid comment posting mechanism with delay scaling as $o(1/\delta)$:

\vspace{0.1in}

\begin{tabular}{|M{0.03\linewidth}M{0.4\linewidth}|M{0.4\linewidth}|}\hline
    & \textbf{Definition of ``Neighboring'' Sets of Comment Arrivals} & \textbf{Impossibility Result} 
    
    \\ \hline\hline 
    
     (1) & Add or remove a batched comment & No \emph{valid} $(\eps, \delta)$-DP posting mechanism for $\eps < \infty, \delta < 1$
    
    \vspace{0.5cm} \\ \hline
    
     (2) & Move a batched comment to another arrival time where it is no longer batched & No \emph{valid} $(\eps, \delta)$-DP posting mechanism for $\eps < \infty, \delta < 1$
    
    \vspace{0.5cm} \\ \hline 
    
      (3) & Move a batched comment by at most $\gap$ units of time to another arrival time where it is no longer batched & For any $\maxDelay \geq 0$, any \emph{valid} $(\eps, \delta)$-DP posting mechanism delays a comment by at least $\maxDelay$ with probability $\geq 1-2\delta \left(\frac{\maxDelay}{\gap} +1\right)$\\\hline
\end{tabular}
\end{theorem}

The proof of the above theorem can be found in Appendix \ref{appendix:impossibility}. Intuitively, we cannot guarantee privacy with definition (1) of neighbors because it would require creating a fake comment since a comment that exists in one input does not exist in the adjacent input. It is not possible to satisfy privacy with definition (2) of neighbors, as a comment could move arbitrarily far in time, requiring infinite delay to be added to comments. For definition (3) of neighbors, we show that we can define a sequence of neighboring inputs such that a comment is shifted $\gap$ units of time in the future on every other input in the sequence. Since the privacy guarantee must hold pairwise between each neighboring input in the sequence, the mechanism can only release comments within time $\maxDelay$ with probability of roughly $\delta \maxDelay / \gap$ in order to make inputs that are $\maxDelay / \gap$ neighbors away from each other in the sequence sufficiently indistinguishable from one another.

Note that even if we considered mechanisms acting on a finite time horizon, the proof above suggests the only mechanism admitted under two-sided DP using definition (2) is the trivial mechanism that releases all comments at the end of the time period: 
\begin{corollary}
Suppose comments are known to arrive only during a finite time horizon $\finiteTime$ after which no more comments will arrive or be posted. Then, any valid posting mechanism that satisfies two-sided $(\eps, \delta)$-DP using Definition (2) of neighbors in Theorem~\ref{thm:impossibility} posts all comments at time $\finiteTime$. 
\end{corollary}

This is both intuitively and formally sufficient for preserving privacy from timing attacks since it eliminates all timing information, but is expensive in terms of delay incurred. In particular, in the peer review setting, releasing all comments simultaneously at the end of the review period eliminates potential for replies and ongoing discussion.

It follows from the impossibility of definition (3) of neighboring sets that there is no valid comment posting mechanism satisfying differential privacy with $\delta = 0$ for this notion of neighbors, since any differentially private mechanism would violate the property that valid comment posting mechanisms eventually release all comments. Additionally, even taking $\delta > 0$, the probability of experiencing a delay longer than $\maxDelay$ only decreases \emph{linearly} in $\delta$ and $\maxDelay$. Typically, $\delta$ is selected to be $o(1/n)$ \cite{dp_textbook}, where $n$ is the database size---in our case, the number of comments in the observed stream. For $\gap=O(1)$, this implies that any mechanism satisfying a two-sided DP guarantee for $\gap$-neighboring inputs (and choosing $\delta = o(1/n)$) has a non-negligible probability of delaying comments by $\Omega(n)$.

\section{Practical Considerations for Implementation}

In this section, we address two important practical considerations to putting into practice our privacy formulation and algorithm. First, we provide theoretically motivated heuristics for setting the parameter $\gap$ in the privacy definition. Second, we give simple extensions to the privacy model and algorithm that allow for handling the realistic setting where batched comments do not arrive all at the same exact time, but rather with a short duration in between. 

\subsection{Setting privacy parameters}
\label{sec:implementation_params}

Recall that our privacy definition includes a parameter $\gap$ that captures what types of inputs can be neighbors. In particular, $\gap$ 
bounds how far apart in time a pair of potentially-batched comments could arrive if batching had \emph{not} taken place. In this section, we provide a heuristic for setting $\gap$ in practice. 
We will argue that a reasonable way to set $\gap$ for a given comment is as a percentile of an empirical distribution of comment inter-arrival times. For example, in a peer-reviewed conference we might set $\gap$ to be the median inter-arrival time of comments at a similar prior conference. Alternatively, larger conferences commonly classify papers into tracks, so $\gap$ could be chosen for each track individually. We provide more examples of setting $\gap$ in practice in our experiments on Wikipedia and Bitcoin in Section~\ref{sec:experiments}.

First, we motivate this heuristic by modeling an adversary conducting a hypothesis test to determine if a comment was batched or not. The privacy parameters $\gap$ and $\eps$ can be chosen based on the desired (in)efficacy of this adversary's test. A natural way to model a privacy attack is to consider an adversary---say, a meta-reviewer who submitted a paper to a conference---who suspects that a comment $\act$ made on their paper may share a reviewer with one of the papers in the set $\anonSet$ of papers they are handling. The adversary conducts a hypothesis test to determine whether the comment they received arrived in a batch with any comment on papers in that set. 
Let $\ts_1$ denote the arrival time of $\act$ and let $\ts_2$ denote the arrival time of the comment in $\anonSet$ that arrives closest in time to $\ts_1$. The adversary knows that if the comments did not arrive in a batch, then they arrived with a gap $\ts_2 - \ts_1$ following some distribution $\distr$ (for instance, this might be the empirical distribution of comment inter-arrival times on the previous day). If the pair of comments does arrive in a batch, the adversary assumes they arrived simultaneously. Thus, the adversary wishes to distinguish between the following hypotheses:
\begin{align*}
   H_0: & \;\; \ts_2 - \ts_1 \sim \distr & (\textit{$\act$ is not batched with any comment in $\anonSet$}) \\ 
   H_1: & \;\; \ts_1 = \ts_2 & (\textit{$\act$ is batched with at least one comment in $\anonSet$})
\end{align*}

The adversary will observe the output of the mechanism and decide to either accept or reject the null hypothesis. If they reject the null hypothesis, they conclude that the comment was batched with a comment in $\anonSet$. Their hypothesis test is defined by ``rejection region'' $R$, or the set of outputs on which the adversary concludes that batching occurred. The quality of a given test is determined by the trade-off between its ``power'' and ``type I error'':
\begin{align*}
\textit{Power} &  = \Pr[\mech(\out) \in R ; H_1] \\ 
\textit{Type I Error} &  = \Pr[\mech(\out) \in R; H_0]
\end{align*}

Similar to prior work on differential privacy~\cite{wasserman_zho\distrU_2008}, \cite{kairouz_composition_2017}, we show that an adversary conducting a hypothesis test to determine if batching occurred will face a poor trade-off between power and type I error given an output of a mechanism that is OSDP with gap $\gap$: 

\begin{proposition}
\label{prop:percentile_osdp}
If a mechanism $\mech$ satisfies $(\eps, \gap)$-OSDP, then for any comment $\act$, set of comments $\anonSet$ arriving with inter-arrival time distribution $\distr$, and any hypothesis test deciding if $\act$ was batched with a consecutively arriving comment in $\anonSet$:
$$\textit{Power} \leq \frac{e^\eps}{\perc(\gap)}  (\textit{Type I Error})$$
where $\perc(\gap) = \Pr[|x| \leq \gap ; x \sim \distr]$ is the CDF of inter-arrival times.
\end{proposition}

The proof of this proposition can be found in Appendix~\ref{app:percentile}. This interpretation of the $(\eps, \gap)$-OSDP guarantee in terms of error rates of an attacker's hypothesis test motivates our heuristic to choose the parameter $\gap$. Previous work on timing attacks~\cite{shmatikov2006timing, levine2004timing} measures the success of attacks in terms of the trade-off between power and type I error. In particular, these works report a single number ``error crossover rate,'' the point at which $\textit{type I error} = 1 - \textit{power}$. 
We envision the system operator (i.e., the entity adding the delay) first specifying a tolerable error crossover rate; for example, consistent with prior work on timing attacks \cite{shmatikov2006timing, levine2004timing}, the operator might choose to tolerate an error crossover rate of $0.25$. 
Next, the system operator should choose a privacy parameter $\epsilon$. 
Since the interpretation of $\epsilon$ is similar to traditional two-sided DP, operators may use common heuristics for selecting $\epsilon$; for instance, our operator might choose $\epsilon=0.8$.
Given these parameters, Proposition \ref{prop:percentile_osdp} shows how to select $g$ to ensure that the desired error crossover rate is satisfied. In our running example, we would choose $g$ to be the 75th percentile of the inter-arrival time distribution. 

\subsection{Handling Non-Simultaneous Batching} 
\label{sec:extension_non_simultaenous}

In our basic model of batching, we make the idealized assumption that all comments in a batch arrive at the same exact clock time. In practice, in many settings, batched actions will not be taken at the \emph{exact} same time, but rather with some short delay between them. For example, it is natural for a Wikipedia editor to spend many minutes working on a revision, so revisions in a single batch may arrive with a few minutes of delay in between. Likewise, reviewers in peer review may comment on papers one after the other, leading to a short delay despite batching.

In this section, we describe a simple extension to our model and algorithm that allows us to handle non-simultaneity in practice. We introduce a new threshold $\batchParam$, below which we consider two comments to have been batched  --- if two comments come from the same reviewer within time $\batchParam$ we consider them to have arrived in a batch. We will assume that $\batchParam < \gap$, as we wish to capture scenarios where batching leads a comment to arrive earlier than it would have without batching. We can capture this scenario by replacing the notion of neighbors in our model with the following:

\begin{definition}[$\gap$-Neighboring Comment Arrival Sets with $\batchParam$-batching]
For $\batchParam < \gap$, a set of comment arrivals $\inpBatch$ is $\gap$-neighboring with $\batchParam$-batching to set $\inp$, if $\inpBatch$ can be obtained from $\inp$ by batching together a pair of comments that arrive separately within $\gap$ time units of one another in $\inp$, moving the later comment to within $\batchParam$ of the earlier comment. Specifically, $\exists (\act, \ts, \paper, \reviewer), (\act', \ts', \paper', \reviewer) \in \inp$ such that $\paper \not = \paper'$, $0 \leq \ts' - \ts \leq \gap$ and $\inpBatch = \inp \setminus \{\act'\} \cup \{(\act', \ts'', \paper', \reviewer)\}$ where $0 \leq |\ts'' - \ts| \leq \batchParam$. 
\end{definition}

We define privacy the same as in Definition~\ref{defn:osdp}, but with this modified notion of $\gap$-neighboring with $\batchParam$-batching. In what follows, we describe how we incorporate this relaxed notion of batching into our algorithm.

First, we propose a simple front-end change that can be employed in conjunction with any mechanism in our algorithmic framework of randomized delay mechanisms (Algorithm~\ref{alg:dp_random_framework}) if we trust users to accurately report when they will engage in batching. The solution is to ask users when they create a comment if they plan on creating more comments on their other papers within the next $\batchParam$ units of time (and hence will generate batched comments). If the user answers affirmatively, then we treat their current comment as well as any subsequent comments they make within $\batchParam$ time units as batched and add delay drawn from $\distrB$ to the batched comments. If not, we add delay from $\distrU$ to the unbatched comments. Here, we take $(\distrB, \distrU)$ to be one-sided $(\eps/2, \gap + \batchParam)$-indistinguishable. Since neighboring inputs can differ on two comments with arrival times at $(\ts, \ts - \batchParam)$ and $(\ts + \batchParam, \ts + \gap)$ respectively, it is now necessary to add noise from $(\eps/2, \gap + \batchParam)$-indistinguishable distributions to preserve privacy by the same reasoning as Theorem~\ref{thm:priv_alg}.

\begin{algorithm}[tb]
\caption{Framework for Handling Non-Simultaneous Batching}
\label{alg:dp_random_extended}
\begin{algorithmic}
   \STATE {\bfseries Input:} privacy parameter $\eps > 0$, maximum gap $\gap > 0$, batching threshold $0 \leq \batchParam < \gap$, noise addition distributions $\distrB$ and $\distrU$
   \FOR{comment arriving at time $\ts$}
    \STATE Hold the comment until time $\ts + \batchParam$. 
    \IF{the same reviewer batched another comment during time $\ts$ to $\ts + \batchParam$}
    \STATE Sample $\delay \sim B(\eps/2, \gap + \batchParam)$ and post the action at time $\ts + \batchParam + \delay$. 
    \ELSE
    \STATE Post the comment at time $\ts + \batchParam + \delay$ where $\delay \sim U(\eps/2, \gap + \batchParam)$
    \ENDIF 
    \ENDFOR 
\end{algorithmic}
\end{algorithm}

In settings where we do not expect users to reliably report that they will batch tasks,  we can use a simple extension to our algorithmic framework, described in Algorithm~\ref{alg:dp_random_extended}, where we delay all comments by an additional $\batchParam$ units of time, using that duration to determine whether or not the comment was batched. The algorithm pays an additional $\batchParam$ in overhead to decide whether a comment was batched or not. Privacy follows by the same reasoning as in Theorem~\ref{thm:priv_alg}. In general, our initial problem formulation captures the most essential features of the problem of preserving privacy in the presence of batching. As we have shown in this section, it is straightforward to extend our model to better capture the properties specific to a given application.

\section{Experiments}
\label{sec:experiments}

We conduct two sets of experiments using publicly available data on Wikipedia article revisions and Bitcoin transactions. 

\subsection{Wikipedia}

In a dataset of revisions on all Wikipedia articles from January 1st to 31st, 2022 obtained from the WikiMedia API, we aggregate over 3.5 million article revisions (after filtering out bot accounts), averaging roughly 80 revisions per minute. Due to the high baseline rate of editing, it would be difficult for an adversary to identify that two revisions are batched without narrowing down the set of possible articles they consider. Therefore, we focus on a subset of Wikipedia revisions within which an adversary tries to link editors. One natural clustering of articles likely to contain batched revisions is by category: each article on Wikipedia is associated with a set of categories capturing the main topics covered. In the following experiments, we analyze articles belonging to the category ``21-st century American Politicians.'' We chose this category because it contains potentially controversial political topics so editors may have privacy concerns in editing these pages. For instance, one news report describes how editors of Donald Trump's Wikipedia page (one of the pages captured in the category) ``are fighting a brutal, petty battle over every word \cite{mak_2019}.'' Additionally, this category receives a large number of revisions per month, yielding a sample size of 13,430 revisions. Among this set of revisions, roughly 20\% were generated in a batch with another revision on a page in the same category (where we consider revisions to be batched if they arrive within $5$ minutes of one another and are made by the same user). The threshold of $5$ minutes captures $92\%$ of pairs of immediately consecutive revisions by a single editor on different articles within this category.

\begin{figure}[tb]
\centering
\begin{subfigure}{.4\textwidth}
  \centering
  \includegraphics[width=.95\linewidth]{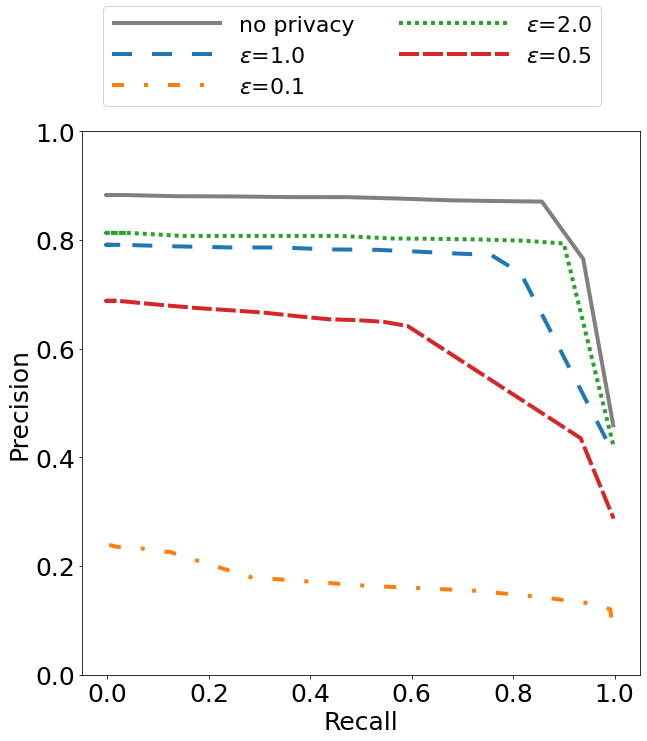}
  \caption{$\gap$ set to 11 minutes}
  \label{fig:wiki_p_25}
\end{subfigure}%
\begin{subfigure}{.4\textwidth}
  \centering
  \includegraphics[width=.95\linewidth]{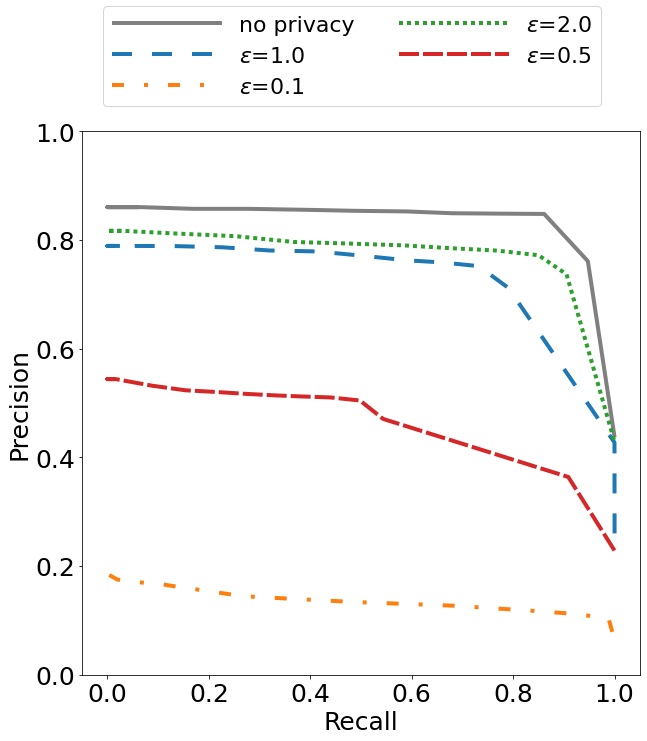}
  \caption{$\gap$ set to 36 minutes}
  \label{fig:wiki_p50}
\end{subfigure}
\caption{Accuracy in linking pairs of Wikipedia article revisions within the category ``21st-century American Politicians'' based on batched timing (averaged over $5$ runs of the randomized privacy mechanism).}
\label{fig:wiki_attack}
\end{figure}

\begin{table}[bt]
    \centering
    \begin{tabular}{|l|cccc|cccc|}
    \hline
         & \multicolumn{4}{c|}{Mean Delay} & \multicolumn{4}{c|}{Maximum Delay} \\\hline
    & $\eps = 0.1$ & $\eps = 0.5$ & $\eps = 1.0$ & $\eps=2.0$ &  $\eps = 0.1$ & $\eps = 0.5$ & $\eps = 1.0$ & $\eps=2.0$ \\ \hline
    (a) $\: \gap = 11$ & 118 & 31 & 20 & 15 & 229 & 54 & 33 & 22 \\\hline
    (b) $\: \gap = 36$ & 343 & 83 & 50 & 35 & 672 & 152 & 88 & 56
 \\\hline
    \end{tabular}
    \caption{Mean and maximum delay (in minutes) added to Wikipedia article revisions within the category ``21st-century American Politicians'' for $\gap$ set to the (a) $25$th and (b) $50$th percentile of the historical inter-arrival distribution.}
    \label{fig:delay_wiki}
\end{table}

While we do not have access to the true identities of editors who use multiple accounts, we can track all revisions made by the same account and identify when this account engages in batching. We simulate an attack where an adversary tries to link revisions to their creator on the basis of timing, while ignoring the usernames of editors. We consider a simple attack model that proves to be quite effective in the absence of any privacy-preserving mechanism. In the attack, the adversary tries to classify each pair of revisions on two different articles as either batched or unbatched.  The adversary chooses a cutoff $\advCutoff \geq 0$: if a pair of revisions are posted within $\advCutoff$ minutes of one another, the adversary classifies the pair as batched and concludes that the comments were made by the same person, and if not, the adversary classifies the pair as unbatched (in which case, the adversary draws no conclusion). When the adversary correctly classifies a batched pair as batched, we call this a true positive, while if the adversary incorrectly classifies an unbatched pair as batched, we call this a false positive. The adversary can trade off between false positives and true positives by choosing the value of $\advCutoff$ accordingly, with higher values of $\advCutoff$ yielding more true positives, but also more false positives, than smaller values of $\advCutoff$. The efficacy of the attack is measured in terms of its precision and recall where $\text{precision} = \frac{\text{number of pairs correctly classified as batched}}{\text{total number of pairs classified as batched}}$ and $\text{recall} = \frac{\text{number of pairs correctly classified as batched}}{\text{total number of pairs that were batched}}$. An effective attack has simultaneously high precision and recall. 

In Figure~\ref{fig:wiki_attack}, we show the precision and recall of this attack under various settings of privacy parameters $\eps$ and $\gap$. We find that attack efficacy is quite high when no privacy mechanism is deployed --- for instance, it is possible to obtain recall of $85\%$ at a precision of $80\%$. We then apply the zero-inflated uniform mechanism (Algorithm~\ref{alg:unif_mechanism}) and measure the reduction in attack efficacy over the ``no privacy'' baseline. We run the mechanism with $\weightB$ set to $1$, as this minimizes the worst-case expected delay added to any single comment in the system. Because batching is not perfectly simultaneous on Wikipedia --- editors take time between making each revision --- we simulate deployment of the user interface extension to Algorithm~\ref{alg:unif_mechanism} described in Section~\ref{sec:extension_non_simultaenous} where $\batchParam=5$ minutes. We set $\gap$ based on the heuristic from Section~\ref{sec:implementation_params} where $\gap$ is a percentile of the inter-arrival distribution of revisions made in the first week of the month. We then simulate deployment of the algorithm over the last three weeks of the month. Using this method, we can set $\gap = 11$ minutes by choosing the 25th percentile or $\gap = 36$ minutes at the 50th percentile. The experiment reveals that precision and recall are significantly improved by use of the mechanism as shown in Figure~\ref{fig:wiki_attack}. 
In terms of delay, Table~\ref{fig:delay_wiki} shows the
mean and maximum delay added to comments. We provide additional results, setting $\gap$ based on the 75th percentile of the inter-arrival distribution in Appendix~\ref{appendix:experiments_wiki}.

Thus we find that Algorithm~\ref{alg:unif_mechanism} renders the privacy attack much less effective while introducing reasonable delay. For instance, taking $\gap = 11$ and $\eps = 0.5$
corresponds to an average delay of roughly $1$ hour $20$ minutes and maximum delay of $2.5$ hours, but makes the attack substantially less accurate: the attack now achieves around $65\%$ recall at $60\%$ precision compared to the non-private baseline which achieves $85\%$ recall at $80\%$ precision. The heuristic attack used in Figure \ref{fig:wiki_attack} may not be optimal for an adversary who has knowledge of the zero-inflated uniform mechanism, but not access to the internal randomness of the mechanism. Identifying an optimal attack is beyond the scope of this work. However, since the same noise distribution $\distrB$ is added to all comments that arrive in a batch, we expect the heuristic attack to perform well in expectation.

\subsection{Bitcoin}

In Bitcoin, we wish to protect against linkage attacks on users of Bitcoin who use multiple addresses to transmit currency to the same recipient address at the same time. We aggregate data of all confirmed transactions broadcast to the Bitcoin peer-to-peer network in the week of August 1, 2022 to August 7, 2022, consisting of approximately 250,000 transactions per day. While we cannot tie different addresses to real-world identities, for the purposes of our experiments, we consider the following proxy: we define a ``batch'' to have occurred when two transactions from different input addresses are sent to the same output address within $1$ minute of one other. This represents a key use-case of our algorithm, wherein a person holding Bitcoin in multiple addresses wishes to draw from these different sources to complete a transfer to a single output address.  After filtering for transactions originating from addresses with unusually high volume of transactions that likely represent cryptocurrency exchanges, there are about 3,000 transactions per day arriving in a batch per our definition, representing 1.2\% of all transactions. 

We consider a privacy attack similar to the linkage attack described in the Wikipedia application. In the Bitcoin setting, an adversary tries to identify whether pairs of transactions arrived in a batch or not. The adversary observes the times at which transactions to the same output address are broadcast to the Bitcoin P2P network and applies a threshold to the time difference between the pair to decide whether the transactions arrived in a batch. In a ``basic'' attack, the adversary uses a single threshold for all transactions. In an ``informed'' attack, we assume the adversary knows the value of $\gap$ that was used by the privacy mechanism for each transaction (which the mechanism may vary by output address) and sets a per-address threshold as a linear function of the $\gap$ used for that address. In incorporating this additional information about the privacy mechanism, the adversary can obtain a better trade-off between false positives and true positives. We measure efficacy of the attack in terms of precision and recall. Since we define batching to occur when multiple inputs are sent to the same output address within $1$ minute of each other, the adversary can observe exactly when batching occurred if no privacy mechanism is deployed and obtain a precision and recall of $100\%$ in identifying whether transactions arrived at the same time or not (recall that in this experiment, we lack ground truth about batched transactions).

\begin{figure}[tb]
    \centering
    \includegraphics[width=0.8\linewidth]{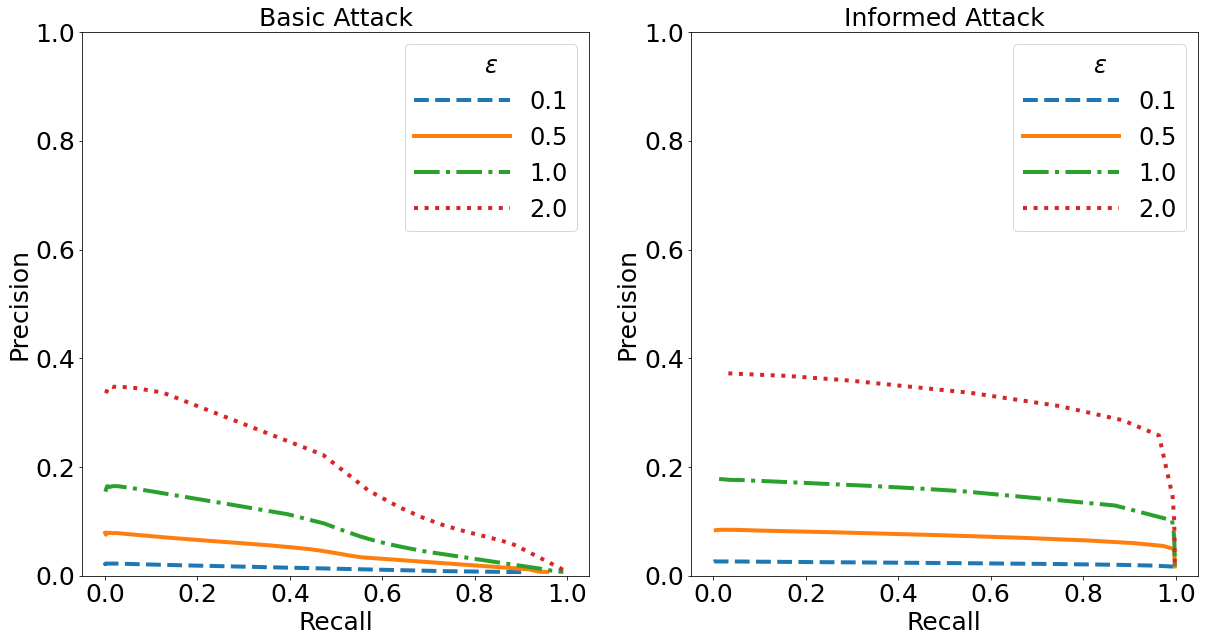}
    \caption{Performance of basic and informed linkage attacks on Bitcoin transactions when $\gap$ is set to the median historical inter-arrival time for an output address.}
    \label{fig:btc_success_50}
\end{figure}

\begin{figure}[tb]
    \centering
    \includegraphics[width=0.5\linewidth]{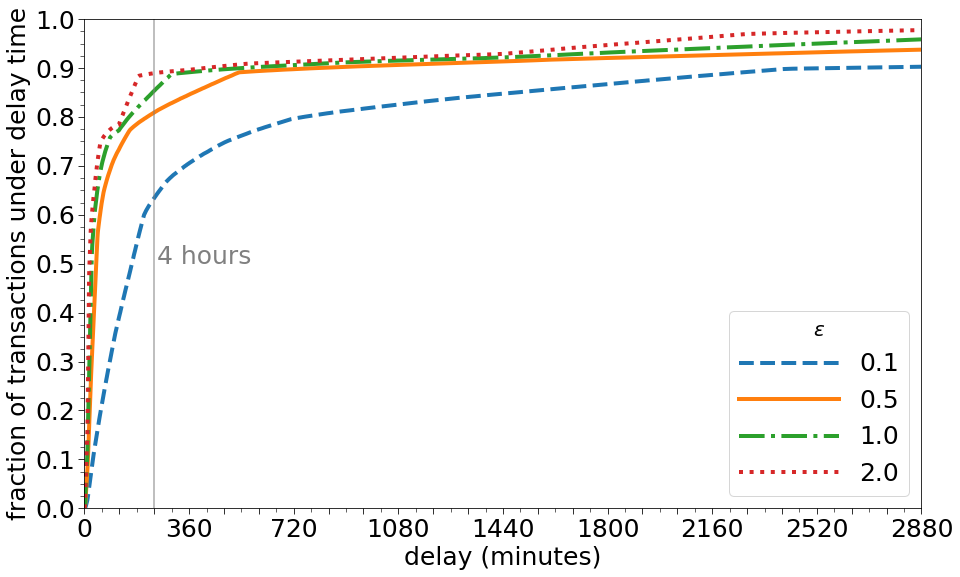}
    \caption{Cumulative distribution of delay added to batched Bitcoin transactions (averaged over 5 trials). Delay is drawn from a privacy-preserving uniform distribution with $\gap$ set to the median of the inter-arrival time of transactions to an output address within the past $7$ days.}
    \label{fig:btc_delay}
\end{figure}

To obscure the timing of transactions, we simulate the zero-inflated uniform mechanism (Algorithm~\ref{alg:unif_mechanism}) to add delay to the time at which transactions are broadcast to the Bitcoin P2P network. In order to select the value of $\gap$, we estimate the inter-arrival distribution of transactions to a given output address in the prior $7$ days and set $\gap$ to a percentile of this distribution. In particular, in this section we use the median of the inter-arrival distribution. In Appendix~\ref{appendix:experiments_btc}, we give additional results for experiments where $\gap$ is set to the 25th and 75th percentile of the inter-arrival distribution. If the output address of a transaction received no other transactions in the prior $7$ days, we set $\gap$ to $10$ minutes, as this is the baseline duration of time a Bitcoin user has to wait for a transaction to be confirmed on the blockchain. Most ($>90\%$) unbatched transactions are sent to output addresses with no recent transaction history, so we use the value of $\gap = 10$ for these transactions. However, roughly $80$\% of batched transactions are sent to output addresses with transaction history. 

The use of Algorithm~\ref{alg:unif_mechanism}, with $\gap$ set per output address, makes it difficult to identify whether transactions to the same output address arrived at the same time. For $\eps = 1$, even the informed attack has precision of only $20\%$ at high recall. The basic attack performs much worse, indicating that an adversary needs to incorporate additional information about baseline inter-arrivals of transactions in order to accurately identify batching. 

This improvement in privacy comes at the expense of added latency. In Figure~\ref{fig:btc_delay}, we show the cumulative density function of delay added to batched Bitcoin transactions averaged over $5$ samples from the privacy-preserving uniform distribution. In general, we can add delay of less than $4$ hours to most transactions. For the setting of $\eps = 1$, the mechanism adds delay of under $2$ hours to $70$\% of transactions. While this is slower than a  Bitcoin transaction when no privacy mechanism is used, it is still substantially faster than many other means of transferring money, like wire transfers. As such, privacy-sensitive users could realistically deploy this algorithm in their Bitcoin wallets to protect the unlinkability of their transactions.  

\section{Discussion}
\label{sec:discussion}

This work introduces the problem of anonymity compromise caused by task batching in pseudonymous forums. We propose defenses and theoretically and empirically establish the efficacy of these solutions.

\paragraph{Global Ordering. }

We find in empirical evaluations of Wikipedia data that the zero-inflated uniform mechanism is likely to release article revisions in a different order than they arrived. In our experiments, at reasonable settings of the privacy parameters, roughly 10\% of revisions were reordered within an article. This can create confusion when there are dependencies between article revisions. A similar problem arises in peer review, where comments may respond to one another. In Appendix~\ref{app:queue_mech}, we discuss a privacy-preserving queue-based mechanism that outputs delayed comments in the same order in which they arrived. While this algorithm does not satisfy the $(\eps, \gap)$-OSDP guarantee, it satisfies a different relaxation of differential privacy. An open question is whether the uniform zero-inflated mechanism can be extended to enforce ordering constraints for an appropriate privacy guarantee.

\paragraph{Partial adoption.}
In actual deployments, many participants may be privacy-insensitive and opt out of additional protections
that preserve anonymity at the cost of increased delay. Our privacy guarantee holds for any pair of events where each event uses the delay mechanism independently of what other users choose to do. So, for a single user who deploys the zero-inflated uniform mechanism on all events, it will be difficult for an adversary to tell whether any pair of their events is batched. However, there may be additional amplification of privacy that comes from widespread usage and permits lower setting of $\gap$ and $\eps$ with the same privacy guarantees in practice. Quantifying the dependence of adoption rate on privacy guarantees is an interesting open question.

\section*{Acknowledgments} 
This research was approved by the CMU Institutional Review Board (IRB). This work was supported in parts by NSF grants CIF: 1763734, 1705007 and RI: 2200410, ONR grant N000142212181, and the Air Force Office of Scientific Research grant FA9550-21-1-0090. The authors gratefully acknowledge the support of the Bill \& Melinda Gates Foundation and the Sloan Foundation. 

\bibliographystyle{abbrv}
\bibliography{bibtex}

~\\~\\~\\~\\

\begin{appendix}

\noindent{\LARGE \bf Appendices}

~\\

In Appendix~\ref{app:pfs}, we present proofs of results that were claimed but not proven in the main text. In Appendix~\ref{app:deanon_attack}, we detail the methods used to measure the prevalence of batching and resulting deanonymization risk in peer review. In Appendix~\ref{app:queue_mech}, we describe an alternative privacy formulation that gives rise to a queue-based mechanism which preserves the order of comment arrivals. Finally, in Appendix~\ref{app:experiments} we give additional empirical results of experiments on Wikipedia and Bitcoin for additional parameter settings not presented in the main text.

\section{Proofs}
\label{app:pfs}

In this section we present proofs of results that were claimed but not proven in the main text. Throughout we will use the following notation to denote element-wise addition and subtraction for a set: for any $S \subseteq \R, t \in \R$, we define $S - t := \{s - t | s \in S\}$.

\subsection{Proof of Theorem~\ref{thm:priv_alg} (Privacy of Random Delay Mechanisms with Indistinguishable Noise-Addition Distributions)}
\label{app:priv_alg}

First, we prove a general necessary and sufficient condition to guarantee $(\eps, \gap)$-OSDP when a mechanism adds independent noise from distributions $\distrB$ and $\distrU$ to batched and unbatched comments respectively.

\begin{lemma}
\label{lem:privacy_condition}
Let $\mech$ be any mechanism that adds independent random delay to comments with delay drawn from distribution $\distrB$ for batched comments and $\distrU$ for unbatched comments. Then, $\mech$ is $(\eps, \gap)$-OSDP if and only if $\forall \out, \out' \in \mathbb{R}$ and $\forall \ts_0 \in [0, \gap]$ it holds that
\begin{align*}
\Pr[\distrB \in \out]\Pr[\distrB \in \out'] \leq e^\eps \Pr[\distrU \in \out]\Pr[\distrU \in \out' - \ts_0].
\end{align*}
\end{lemma}

\begin{proof}
First, let $\mech$ be any $(\eps, \gap)$-OSDP mechanism adding independent random delay to comments with delay drawn from distributions $\distrB$ and $\distrU$. Suppose for the sake of contradiction that there exists some $\out, \out' \in \R$ and $\ts_0 \in [0, \gap]$ such that $\Pr[\distrB \in \out]\Pr[\distrB \in \out'] > e^\eps \Pr[\distrU \in \out]\Pr[\distrU \in \out' - \ts_0]$. Let $\inp$ and $\inpBatch$ be $\gap$-neighboring inputs differing in the arrival time of a single comment. In $\inpBatch$, a pair of comments $\act$ and $\act'$ arrive in a batch at time $0$.  In $\inp$, comment $\act$ arrives unbatched at time $0$ and comment $\act'$ arrives unbatched at time $\ts_0$. All other comments arrive at the same times in $\inp$ and $\inpBatch$. Let $O$ denote the set of possible outputs where $\act$ is posted at a time in $\out$ and $\act'$ is posted at a time in $\out'$ and all other comments are posted at any time in $\R$. Then, since delay is added independently to each comment: $\Pr[\mech(\inpBatch) \in O] = \Pr[\distrB \in \out] \Pr[\distrB \in \out']$ and $\Pr[\mech(\inp) \in O] = \Pr[\distrU \in \out] \Pr[\distrU \in \out' - \ts_0]$. However, by the initial assumption $\Pr[\distrB \in \out]\Pr[\distrB \in \out'] > e^\eps \Pr[\distrU \in \out]\Pr[\distrU \in \out' - \ts_0]$ contradicting the $(\eps, \gap)$-OSDP of $\mech$.

Now, we prove the other direction. Let $\mech$ be any mechanism adding independent random delay to comments with delay drawn from distributions $\distrB$ and $\distrU$ such that $\forall \out, \out' \in \mathbb{R}$ and $\forall \ts_0 \in [0, \gap]$ it holds that
\begin{align}
    \Pr[\distrB \in \out]\Pr[\distrB \in \out'] \leq e^\eps \Pr[\distrU \in \out]\Pr[\distrU \in \out' - \ts_0].
\end{align} Note that taking $\out = \mathbb{R}$, $\Pr[\distrB \in \out] = \Pr[\distrU \in \out] = 1$ so it must hold that $\forall \out' \in \mathbb{R}, \ts_0 \in [0,\gap]$
\begin{align}
    \Pr[\distrB \in \out'] \leq e^{\eps} \Pr[\distrU \in \out' - \ts_0].
\end{align}

Let $\inp$ and $\inpBatch$ be any $\gap$-adjacent comment arrival sets. Let $\act, \act'$ denote the pair of comments that arrive in a batch together in $\inpBatch$ but do not arrive in a batch together in $\inp$. In $\inp$, the two comments both arrive at time $\ts$, while in $\inpBatch$ comment $\act$ arrives at time $\ts$ and comment $\act'$ arrives at time $\ts + \ts_0$ with $\ts_0 \in [0, \gap]$ by the definition of $\gap$-adjacency. All other comments arrive at the same time in $\inp$ and $\inpBatch$. Let $o$ and $o'$ denote the randomized times at which the mechanism $\mech$ releases comments $\act$ and $\act'$ respectively.

Let $O$ be any set of possible outputs of the mechanism during time horizon $\finiteTime$ and let $S$ denote the values of $o$ in $O$ and $S'$ the values of $o'$ in $S'$. Then, because $\mech$ adds noise independently to each comment and all comments other than $\act$ and $\act'$ are equivalent in $\inp$ and $\inpBatch$, the probabilities factor as \begin{align*}
\Pr[\mech(\inpBatch) \in O] = k  \Pr[o \in \out ; \inpBatch]  \Pr[o' \in \out' ; \inpBatch]  \\
\text{and} \; \Pr[\mech(\inp) \in O] = k  \Pr[o \in \out ; \inpBatch]  \Pr[o' \in \out' ; \inpBatch] 
\end{align*} where $k$ captures the probability that all comments other than $\out$ and $\out'$ are posted at post times in the set of outputs $O$.

Consider two cases for the size of the batch in which $\act$ and $\act'$ arrive in $\inpBatch$. First, suppose the batch has $2$ comments. Then, on input $\inp$, both comments are unbatched, so for some $t_0 \in [0, \gap]$ and $\delay_\distrU, \delay'_\distrU \simiid \distrU$: $o = \ts + \delay_\distrU$ and $o' = \ts + \ts_0 + \delay'_\distrU$ . On input $\inpBatch$, $o = \ts + \delay_\distrB$ and $o' = \ts + \delay'_\distrB$ where $\delay_\distrB, \delay'_\distrB \simiid \distrB$. Therefore, \begin{align*}
\Pr[\mech(\inpBatch) \in O] = k \Pr[\ts + \delay_\distrB \in \out]  \Pr[\ts + \delay'_\distrB \in \out'] = k \Pr[\distrB \in \out - \ts] \Pr[\distrB \in \out' - \ts]    \\ 
\text{and} \; \Pr[\mech(\inp) \in O] =  k \Pr[\ts + \delay_\distrU \in \out]  \Pr[\ts + \ts_0 + \delay'_\distrU \in \out'] = k \Pr[\distrU \in \out - \ts] \Pr[\distrU \in \out' - \ts - \ts_0]. \\ 
\end{align*}
So, by Inequality (1) we have $\Pr[\mech(\inpBatch) \in O] \leq e^\eps \Pr[\mech(\inp) \in O]$. In the case where the batch containing $\act$ and $\act'$ has only two comments, the probability of the output on input $\inpBatch$ remains the same, but on $\inp$, $o = \ts + \delay_\distrB$, since comment $\act$ is still treated as batched, so 
$\Pr[\mech(\inp) \in O] = k \Pr[\distrB \in \out - \ts] \Pr[\distrU \in \out' - \ts - \ts_0]$ and by Inequality (2), $\Pr[\mech(\inpBatch) \in O] \leq e^\eps \Pr[\mech(\inp) \in O]$ so $\mech$ is $(\eps, \gap)$-OSDP.
\end{proof}

Let $\mech$ be any mechanism adding independent delay from $\distrB$ and $\distrU$ to batched and unbatched comments respectively, where $\distrB$ and $\distrU$ are $(\eps/2, \gap)$-one-sided indistinguishable distributions. Consider any $\out, \out' \in \R$ and $\ts_0 \in [0,\gap]$. Then, by indistinguishability $\Pr[\distrB \in \out] \leq e^{\eps/2} \Pr[\distrU \in \out]$ and $\Pr[\distrB \in \out' ] \leq e^{\eps/2} \Pr[\distrU \in \out - \ts_0]$, so $\Pr[\distrB \in \out]\Pr[\distrB \in \out'] \leq e^{\eps} \Pr[\distrU \in \out] \Pr[\distrU \in \out' - \ts_0]$. Applying Lemma~\ref{lem:privacy_condition} we conclude that $\mech$ is $(\eps, \gap)$-OSDP completing the proof.

\subsection{Proof of Theorem~\ref{thm:privacy_distributions} (Privacy-preserving distributions)}
\label{app:pf_indistinguish}

First, note that by the definition of one-sided indistinguishability (Definition~\ref{defn:indisting_distr}), if $\distrB$ and $\distrU$ have probability density functions $b$ and $u$ respectively then $\distrB$ and $\distrU$ are $(\eps, \gap)$-one-sided indistinguishable if and only if $\frac{b(\ts)}{u(\ts-\ts_0)} \leq e^\eps \;\; \forall \ts \geq 0, \ts_0 \in [0, \gap]$ for which $b(\ts) > 0$. So,

\begin{enumerate}[label=(\arabic*)]
    \item Exponential: for any $\ts < \gap$, $b(\ts) = 0$ while for any $\ts \geq \gap, \ts_0 \in [0, \gap]$ it holds that $\frac{b(\ts)}{u(\ts-\ts_0)} = \frac{\exp\{- \eps (\ts - \gap) / \gap\}}{\exp\{-\eps (\ts - \ts_0) / \gap)\}} \leq  \frac{\exp\{- \eps (\ts - \gap) / \gap\}}{\exp\{-\eps \ts / \gap)\}} = e^{\eps}$.
    \item Staircase: by indistinguishability of the staircase distribution proven in \cite{staircase}.
    \item Uniform: for any $\ts \in [\gap, \frac{1}{1-e^{-\eps}} \gap]$ , $\ts_0 \in [0, \gap]$ we have that $\frac{b(\ts)}{u(\ts-\ts_0)} = \frac{(1-e^{-\eps}) / (e^{-\eps} \gap)}{(1-e^{-\eps}) / \gap} = e^{\eps}$ and $b(\ts) = 0$ for all other values of $\ts$ so $b(\ts) = 0 \leq e^\eps u(\ts - \ts_0)$ for all other values of $\ts$. 
    \item Zero-inflated uniform: for any closed interval $[a,b] \subset [0, \gap)$ or $[a,b] \subset [\frac{\pUnif}{\pUnif - e^{-\eps}} \gap, \infty)$, we have that $\Pr[\distrB \in S] = 0$. For any interval $[a,b] \subseteq [\gap, \frac{\pUnif}{\pUnif - e^{-\eps}} \gap ]$, we have that $\Pr[B \in [a,b]] = (b-a) \frac{p - e^{-\eps}}{ e^{-\eps} \gap}$ while for any $\ts_0 \in [0, \gap] $ we have that $\Pr[U \in [a - \ts_0, b - \ts_0]] \geq (b-a) \frac{p - e^{-\eps}}{ \gap}$ so the ratio $\frac{\Pr[B \in S]}{\Pr[U \in S - t_0]}$ is bounded by $e^{\eps}$ for any measurable set $S$. 
\end{enumerate}

\subsection{Proof of Theorem~\ref{thm:distr_params} (Optimal choice of parameters for exponential and staircase distributions)}
\label{appendix:exp_staircase_proof}

First, in the following two lemmas we argue that the offset terms $\distrOffset_\distrB$ and $\distrOffset_\distrU$ must be set to $\distrOffset_\distrB = \gap$ and $\distrOffset_\distrU = 0$ in any expectation-minimizing pair of distributions that guarantees privacy. 

\begin{lemma}
\label{lem:offsets_tight}
Let $\distrB, \distrU$ be non-negative noise-addition distributions that guarantee $(\eps, \gap)$-OSDP when used in Algorithm~\ref{alg:dp_random_framework} where $\distrB = \distrOffset_\distrB + \distrB_0$ and $\distrU = \distrOffset_\distrU + \distrU_0$ for constants $\distrOffset_\distrB, \distrOffset_\distrU > 0$ and random variables $\distrB_0$ and $\distrU_0$ with support $[0, \infty)$. Then, $\distrOffset_\distrB - \distrOffset_\distrU \geq \gap$.
\end{lemma}

\begin{proof}
By Lemma~\ref{lem:privacy_condition}, in order for privacy to hold it must be that $\forall \out, \out' \subseteq \text{support}(\distrB), \ts_0 \in [0,\gap]$:
\begin{align*}
  \frac{\Pr[\distrB \in \out]  \Pr[\distrB \in \out' ]}{\Pr[\distrU \in \out]  \Pr[\distrU \in \out' - \ts_0]} \leq e^\eps.
\end{align*}
so taking $\ts_0 = \gap$ and $\distrB = \distrOffset_\distrB + \distrB_0$ and $\distrU = \distrOffset_\distrU + \distrU_0$ we have that
\begin{align*}
  \frac{\Pr[\distrB_0 \in \out - \distrOffset_\distrB]  \Pr[\distrB_0 \in \out' - \distrOffset_\distrB ]}{\Pr[\distrU_0 \in \out - \distrOffset_\distrU]  \Pr[\distrU_0 \in \out' - (\distrOffset_\distrU + \gap)]} \leq e^\eps.
\end{align*}

Suppose for the sake of contradiction that $\distrOffset_\distrB < \distrOffset_\distrU + \gap$. Then, taking $\out = \out' = [\distrOffset_\distrB, \distrOffset_\distrU + \gap)$ we have that
$\Pr[\distrB_0 \in \out' - \distrOffset_\distrB] = \Pr[\distrB_0 \in [0, \distrOffset_\distrU + \gap - \distrOffset_\distrB)] > 0$, but $\Pr[\distrU_0 \in \out' - (\distrOffset_\distrU + \gap)] = \Pr[\distrU_0 \in [\distrOffset_\distrB - \distrOffset_\distrU - \gap , 0)] = 0$ so the likelihood ratio is unbounded yielding a contradiction.  
\end{proof}

\begin{lemma}
\label{lem:offset_gap}
Let $\distrB, \distrU$ be non-negative noise-addition distributions that guarantee $(\eps, \gap)$-OSDP when used in Algorithm~\ref{alg:dp_random_framework} where $\distrB = \distrOffset_\distrB + \distrB_0$ and $\distrU = \distrOffset_\distrU + \distrU_0$ for constants $\distrOffset_\distrB, \distrOffset_\distrU > 0$ and random variables $\distrB_0$ and $\distrU_0$ where either $\distrB_0$ and $\distrU_0$ are both exponential random variables or staircase random variables. Then, $\distrB' = \gap + \distrB_0$ and $\distrU'= 0 + \distrU_0$ guarantee $(\eps, \gap)$-OSDP when used in Algorithm~\ref{alg:dp_random_framework}.
\end{lemma}

\begin{proof}
Let $\distrB = \distrOffset_\distrB + \distrB_0$ and $\distrU = \distrOffset_\distrU + \distrU_0$ be distributions that satisfy $(\eps, \gap)$-OSDP when used in Algorithm~\ref{alg:dp_random_framework}. First, define $\distrB' = \distrB - \distrOffset_\distrU$ and $\distrU' = \distrU - \distrOffset_\distrU$. Note that by Lemma~\ref{lem:offsets_tight}, $\distrOffset_\distrB > \distrOffset_\distrU$ so $\distrB'$ is still non-negative. Since we both random variables are shifted by the same constant offset, $\distrB'$ and $\distrU'$ still satisfy the sufficient condition to guarantee privacy in Lemma~\ref{lem:privacy_condition}. Now, suppose that $\distrOffset_\distrU = 0$ and $\distrOffset_\distrB > \gap$. Note that both the staircase distribution or the exponential distribution have monotonically decreasing probability density functions above $0$ so $\Pr[\distrB_0 \in \out - \distrOffset_\distrB] \geq \Pr[\distrB_0 \in \out - \gap]$. Therefore, setting $\distrB' = \gap + \distrB_0$ the sufficient condition for privacy in Lemma~\ref{lem:privacy_condition} still holds.
\end{proof}

Now, taking $\distrB = \gap + \distrSingle$ and $\distrU = \distrSingle$, by Lemma~\ref{lem:privacy_condition}, distribution $\distrSingle$ must satisfy the condition that $\forall \out, \out' \subseteq [\gap, \infty), \ts_0 \in [0,\gap)$
\begin{align*}
  \frac{\Pr[\distrSingle \in \out - \gap]  \Pr[\distrSingle \in \out' - \gap]}{\Pr[\distrSingle \in \out]  \Pr[\distrSingle \in \out' - \ts_0]} \leq e^\eps.
\end{align*}
Taking $\ts_0 = 0$ and $\out = \out'$, the privacy constraint requires that $\forall \out \subseteq [\gap, \infty)$:
\begin{align}
  \frac{\Pr[\distrSingle \in \out - \gap]}{\Pr[\distrSingle \in \out]} \leq e^{\eps/2}.
\end{align}
So, if $\distrSingle$ is an exponential distribution with rate parameter $\lambda$, then $\forall x \in [\gap, \infty)$
\begin{align*}
\frac{\lambda \exp\{-\lambda (x - \gap) \}}{\lambda \exp\{-\lambda x\}} = \exp \{ \lambda \gap \} \leq \exp\{\eps/2\}
\end{align*}
Then, the expectation of $\distrB$ and $\distrU$ is minimized by taking $\lambda = \frac{\eps}{2\gap}$.

If $\distrSingle$ is a staircase distribution, it follows from the proof of optimality in $\cite{staircase}$ (Theorem 4), that the staircase distribution with parameters $(\eps', \Delta, \gamma)$ set to $\eps' = \eps/2$, $\Delta = \gap$ and $\gamma = \frac{1}{1+e^{\eps/2}}$ respectively is optimal in minimizing the expectation of $\distrSingle$ while respecting Inequality (3) completing the proof.  

\subsection{Proof of Theorem~\ref{thm:unif_opt} (Pareto frontier)}
\label{appendix:optimal_proof}

The proof will proceed in three parts. First, in Section~\ref{sec:restrict_indisting} we argue that we can restrict attention to distributions $\distrB$ and $\distrU$ such that $\distrB$ and $\distrU$ are $(\eps/2, \gap)$-one-sided indistinguishable. Second, in Section~\ref{sec:pareto_proof_sec}, we prove that among $(\eps/2, \gap)$-one-sided indistinguishable distributions any Pareto optimal pair of distributions must be zero-inflated uniform distributions. Finally, in Section~\ref{sec:param_choice} we derive the optimal choice of parameters of the zero-inflated uniform distribution as a function of privacy parameters $\eps, \gap$ and choice of weighted utility function $\weightB \E[\distrB] + (1-\weightB) \E[\distrU]$.

\subsubsection{Restricting attention to $(\eps/2, \gap)$-one-sided indistinguishable distributions}
\label{sec:restrict_indisting}

We being by arguing that we can restrict attention to finding optimal noise addition distributions $\distrB, \distrU$ such that $\distrB$ and $\distrU$ are $(\eps/2, \gap)$-one-sided indistinguishable distributions (Definition~\ref{defn:indisting_distr}) and then use these distributions within the framework of Algorithm~\ref{alg:dp_random_framework} to design an optimal mechanism.

\begin{lemma}
Let $\mech$ be any valid $(\eps, \gap)$-OSDP comment posting mechanism that adds independent noise drawn from distributions $\distrB$ and $\distrU$ to batched and unbatched comments respectively. Then, $(\distrB, \distrU)$ must be $(\eps, \gap)$-one-sided indistinguishable.
\end{lemma}

\begin{proof}
By Lemma~\ref{lem:privacy_condition}, in order for privacy to hold for a mechanism that adds independent noise drawn from distributions $\distrB$ and $\distrU$ respectively, it must be that $\forall \out, \out' \subseteq \R$ such that $\Pr[\distrB \in \out] > 0$ and $\Pr[\distrB \in \out'] > 0$ and $\forall \ts_0 \in [0,\gap]$:
\begin{align*}
  \frac{\Pr[\distrB \in \out]  \Pr[\distrB \in \out' ]}{\Pr[\distrU \in \out]  \Pr[\distrU \in \out' - \ts_0]} \leq e^\eps.
\end{align*}
Then, taking $\out=\R$, $\frac{\Pr[\distrB \in \out]}{\Pr[\distrU \in \out]} = 1$, so $\frac{\Pr[\distrB' \in \out' ]}{\Pr[\distrU' \in \out' - \ts_0]} \leq e^\eps \; \forall \out' \subseteq \R, \ts_0 \in [0,\gap]$.
\end{proof}

Note that Algorithm~\ref{alg:dp_random_framework} (of which optimal Algorithm~\ref{alg:unif_mechanism} is an instance) adds noise from $(\eps/2, \gap)$-indistinguishable distributions, which is a stronger condition than requiring $(\eps, \gap)$-indistinguishable distributions. We will prove below (in Lemma~\ref{lem:non_increasing_step}) that for any Pareto optimal $(\eps, \gap)$-indistinguishable distributions $(\distrB, \distrU)$, $\distrU$ must be monotonically non-increasing above $0$. It follows that the distributions must be $(\eps/2, \gap)$-indistinguishable in order for Algorithm~\ref{alg:dp_random_framework} to be $(\eps, \gap)$-OSDP:

\begin{lemma}
Let $\mech$ be any valid $(\eps, \gap)$-OSDP comment posting mechanism that adds independent noise drawn from distributions $\distrB$ and $\distrU$ to batched and unbatched comments respectively where $\distrU$ is monotonically non-increasing (above $0$). Then, $(\distrB, \distrU)$ must be $(\eps/2, \gap)$-one-sided indistinguishable (Definition~\ref{defn:indisting_distr}).
\end{lemma}

\begin{proof}
By Lemma~\ref{lem:privacy_condition}, it must be that $\forall \out, \out' \subseteq \R$ such that $\Pr[\distrB \in \out] > 0$ and $\Pr[\distrB \in \out'] > 0$ and $\forall \ts_0 \in [0,\gap]$:
\begin{align*}
  \frac{\Pr[\distrB \in \out]  \Pr[\distrB \in \out' ]}{\Pr[\distrU \in \out]  \Pr[\distrU \in \out' - \ts_0]} \leq e^\eps.
\end{align*}
Taking $\out = \out'$ and $\ts_0 = 0$ gives $\frac{\Pr[\distrB \in \out]}{\Pr[\distrU \in \out]} \leq e^{\eps/2} \; \forall \out \subseteq \R$. Since $\distrU$ is non-increasing, $\Pr[\distrU \in \out - \ts_0] \geq \Pr[\distrU \in \out]$ for $\ts_0 \geq 0$, so $\frac{\Pr[\distrB \in \out]}{\Pr[\distrU \in \out - \ts_0]} \leq e^{\eps / 2}$ as well and $\distrB$ and $\distrU$ are $(\eps/2, \gap)$-one-sided indistinguishable. 
\end{proof}

\subsubsection{Pareto optimal distributions}
\label{sec:pareto_proof_sec}

The main portion of this proof characterizes Pareto optimal distributions $(\distrB, \distrU)$ such that $\distrB$ and $\distrU$ are $(\eps, \gap)$-one-sided indistinguishable. From Section \ref{sec:restrict_indisting}, we can then choose $(\eps/2, \gap)$-indistinguishable distributions for use in Algorithm~\ref{alg:dp_random_framework} to obtain an optimal mechanism.

Let $\mathcal{P}_{\eps, \gap}$ denote the set of all pairs of $(\eps, \gap)$-one-sided indistinguishable distributions (Definition~\ref{defn:indisting_distr}). To derive the Pareto frontier of $\mathcal{P}_{\eps, \gap}$, we follow the high-level approach of~\cite{staircase}, which derives the optimal \emph{two-sided} differential privacy noise-addition distribution. The proof proceeds by showing that if $\distrB$ and $\distrU$ are $(\eps, \gap)$-one-sided indistinguishable distributions added to batched and unbatched comments respectively, then:
\begin{enumerate}
    \item $\distrB$ and $\distrU$ can be approximated arbitrarily well by a random variable defined by an appropriately chosen piece-wise constant probability density function.
    \item We derive various properties of Pareto optimal $\distrB$ and $\distrU$ by showing that we can shift probability mass around in the piece-wise constant approximations to $\distrB$ and $\distrU$, such that we decrease expected delay while maintaining indistinguishability. In particular, we show that $\distrB$ must place $0$ probability mass below $\gap$ and any Pareto optimal $\distrB$ must be monotonically non-increasing above $\gap$. We show that $\distrU$ is uniquely defined by $\distrB$ to put as little probability mass at each point as possible to maintain indistinguishability with $\distrB$ and put any excess probability mass at $0$. We then prove that these properties imply that the zero-inflated uniform distribution is Pareto optimal. 
\end{enumerate}

For a random variable $X$ and for any positive integer $i > 0$, define a random variable $X_i$ that approximates $X$ where $X_i$ has probability density function $f^{(X)}_i(\cdot)$ with constant density over intervals of length $\frac{\gap}{i}$:
\begin{align}
   f^{(X)}_i(t) = \begin{cases}
    \frac{\Pr\left(X \in \left[k\tfrac{\gap}{i}, (k+1)\tfrac{\gap}{i}\right)\right)}{\tfrac{\gap}{i}} & \text{if }  t \in[k \tfrac{\gap}{i}, (k+1) \frac{\gap}{i}) \text{ for } k \in \N \\ 
    0 & \text{if } t < 0.
   \end{cases}   
\end{align} Given $(\distrB, \distrU) \in \mathcal{P}_{\eps, \gap}$, for any positive integer $i > 0$ define $(\distrB_i, \distrU_i)$ to be the random variables with probability density functions $f^{(B)}_i(\cdot)$ and $f^{(U)}_i(\cdot)$ taken to be the step-function approximations to $\distrB$ and $\distrU$ defined in Equation (3). Since the probability density function of each distribution is piece-wise constant, we define a ``probability density sequence'' of each distribution ($\{b^{(i)}_k\}_{k=0}^\infty$ and $\{u^{(i)}_k \}_{k=0}^\infty$ respectively) to be the sequence of values of the pdf for each constant interval of length $\gap / i$. For instance, $b^{(i)}_0$ corresponds to the constant probability density for values in range $0$ to $\gap / i$ while $b^{(i)}_i$ corresponds to the probability density over range $\gap$ to $(\gap+1) / i$. 

\begin{lemma}[Piecewise Constant Approximation]
\label{lem:piecewise}
For any $B,U \in \mathcal{P}_{\eps, \gamma}$ and $i \in \N$ the following properties hold for piece-wise constant approximations $(\distrB_i, \distrU_i)$ to $(\distrB, \distrU)$ with probability density functions $f^{(B)}_i$ and $f^{(U)}_i$ respectively:
\begin{enumerate}[label=(\roman*)]
    \item (Valid Probability Distributions) $f^{(B)}_i$ and $f^{(U)}_i$ are non-negative functions that integrate to $1$.
    \item (Indistinguishability) $(\distrB_i, \distrU_i) \in \mathcal{P}_{\eps, \gamma}$.
    \item (Convergence of Expected Value) $\lim_{i \to \infty} (\E[\distrB_i], \E[\distrU_i]) =  (\E[\distrB], \E[\distrU])$.
\end{enumerate}
\end{lemma}

\begin{proof} We prove each claim separately:

\begin{enumerate}[label=(\roman*)]
\item For any random variable $X$ with approximation $X_i$ we have $$\int_{0}^{\infty} f_i^{(X)}(t) dt =  \sum_{k=0}^{\infty} \int_{\left[\tfrac{k\gap}{i}, \tfrac{(k+1)\gap}{i}\right)} f_i^{(X)}(t) dt = \sum_{k=0}^{\infty} \Pr(X \in [\tfrac{k\gap}{i}, \tfrac{(k+1)\gap}{i})) = 1.$$
\item For any $\ell \in \{0,\ldots,\min(i, k) \}$: 
\begin{align*}
\frac{\bSeq}{u^{(i)}_{k-\ell}} = \frac{ \Pr(\distrB \in [k\gap/i, (k+1)\gap/i])}{\Pr(\distrU \in [(k- \ell)\gap/i, (k-\ell+1)\gap/i])} \leq e^{\eps}
\end{align*}
 by indistinguishability of $\distrB$ and $\distrU$ and since the interval in the denominator is the same length interval as the numerator shifted by at most $g$ to the left. Hence, for any $t \in [0, \infty), t_0 \in [0,\gap]$: $\frac{\distrB_i(t)}{\distrU_i(t - t_0)} \leq e^\eps$ so $(\distrB_i, \distrU_i) \in \mathcal{P}_{\eps, \gap}$.
\item In~\cite{staircase} Lemma 19 in Appendix B proves that for any random variable $X$ and approximation $X_i$ defined as above, $\lim_{i\to \infty} \E[X_i] = \E[X]$. So, $\lim_{i \to \infty} (\E[\distrB_i], \E[\distrU_i]) = (\lim_{i \to \infty} \E[\distrB_i], \lim_{i \to \infty} \E[\distrU_i]) = (\E[\distrB], \E[\distrU])$.
\end{enumerate}
\end{proof}

It follows from from parts (ii) and (iii) of Lemma~\ref{lem:piecewise} that
\begin{corollary}
\label{cor:weighted_sum}
For any fixed $[w_\distrB, w_\distrU] \in [0,1]^2$ with $w_\distrB + w_\distrU = 1$: \begin{align*}\inf_{(\distrB_i, \distrU_i) \in \bigcup_{i=1}^{\infty} \mathcal{P}^{(i)}_{\eps,\gap}} w_\distrB \E[\distrB_i] + w_\distrU \E[\distrU_i] = \inf_{(\distrB, \distrU) \in \mathcal{P}_{\eps, \gap}} w_\distrB \E[\distrB] + w_\distrU \E[\distrU].\end{align*}
\end{corollary}

Now, we show that deriving the Pareto frontier of $\mathcal{P}_{\eps, \gap}$ is equivalent to optimizing any weighted sum of $\E[\distrB]$ and $\E[\distrU]$ because the feasible region is convex. Therefore, we can focus on characterizing $\mathcal{P}^{(i)}_{\eps, \gap}$ that are optimal for the weighted sum objective and take the limit as $i \to \infty$ to derive the entire Pareto frontier of $\mathcal{P}_{\eps, \gap}$. 

\begin{lemma}
\label{lem:frontier_weighted_sum}
If $(\distrB, \distrU)$ is Pareto optimal, then it minimizes some weighted sum of $\E[\distrB]$ and $\E[\distrU]$:  $\exists (w_\distrB, w_\distrU) \in [0,1]^2$ with $w_\distrB + w_\distrU = 1$ such that $$(\distrB, \distrU) \in \argmin_{(B',U') \in \mathcal{P}_{\eps, \gap}} w_\distrB\E[\distrB'] + w_\distrU\E[\distrU'].$$
\end{lemma}

\begin{proof}
We argue that the feasible region $\{(\E[\distrB], \E[\distrU]) \;|\; (\distrB, \distrU) \in \mathcal{P}_{\eps, \gap}\}$ is convex. Take $(\distrB_1,\distrU_1) , (\distrB_2, \distrU_2) \in \mathcal{P}_{\eps, \gap}$ with $E_1 = (\E[\distrB_1], \E[\distrU_1])$ and $E_2 = (\E[\distrB_2], \E[\distrU_2])$. For any $p \in [0,1]$ define random variable $\distrB_3$ to be the random variable that samples $\distrB_1$ with probability $p$ and $\distrB_2$ with probability $(1-p)$ and define $\distrU_3$ accordingly with respect to $\distrU_1,\distrU_2$. Then, for any measurable set $S \subseteq \R$, $$\Pr[\distrB_3 \in S] = p \Pr[\distrB_1 \in S] + (1-p) \Pr[\distrB_2 \in S] \leq e^{\eps} p \Pr[\distrU_1 \in S] + (1-p) e^\eps \Pr[\distrU_2 \in S] = e^{\eps} \Pr[\distrU_3 \in S]$$ so $(\distrB_3, \distrU_3) \in \mathcal{P}_{\eps,\gap}$ and have expectations $p  E_1 + (1-p)  E_2$. Then, we apply the fact that all points in the Pareto frontier of a convex feasible region are solutions to a weighted sum optimization problem (see, for instance, Boyd~\cite[Chapter 4.7]{boyd2004convex}).
\end{proof}

\subsubsection*{Properties of Pareto Optimal $\distrB_i, \distrU_i$:}
\label{sum:properties}

Below, we establish the following properties of any Pareto optimal $(\distrB_i, \distrU_i) \in \mathcal{P}_{\eps, \gap}^{(i)}$ for any $i \in \N$ with probability density sequences $\{\bSeq\}_{k=0}^\infty$ and $\{\uSeq\}_{k=0}^\infty$ respectively:
\begin{enumerate}[label=(\arabic*)]
    \item $\bSeq=0$ for all $k < i$, $\{\bSeq\}$ is non-increasing for all $k \geq i$, and $\bSeq$ is bounded by $\bSeq \leq \frac{1 - e^{-\eps}}{e^{-\eps} \gap}$ for all $k$.
    \item $\uSeq$ is fully determined by choice of $\bSeq$, that is, $\uSeq = e^{-\eps} b^{(i)}_i$ for all $k \in [1, i)$, $\uSeq = e^{-\eps} \bSeq$ for all $k \geq i$, and $u^{(i)}_0 = \frac{i}{\gap} ( 1 - e^{-\eps} - \tfrac{(i-1)}{i} \gap e^{-\eps} b^{(i)}_i) \geq \frac{1-e^{-\eps}}{\gap}$.
\end{enumerate}

\begin{lemma}[Support of $\distrB$]
\label{lem:support_\distrB}
Let $\distrB$ and $\distrU$ be any $(\eps, \gap)$-one-sided indistinguishable distributions. Then, $\Pr[\distrB < \gap] = 0$.
\end{lemma}

\begin{proof}
By indistinguishability $\frac{\Pr(\distrB \in [0,\gap))}{\Pr(\distrU \in [-\gap, 0))} \leq e^\eps$, but by non-negativity, $\Pr(\distrU \in [-\gap, 0)) = 0$. So, $\Pr(\distrB \in [0, \gap)) = 0$.
\end{proof}

Note that by definition of $\distrB_i$, the above lemma proves that $\bSeq=0$ for all $k < i$, since any interval below $i$ corresponds to the density of the random variable at a value below $\gap$. 

\begin{lemma}[Upper bound on $b$]
\label{lem:bound_\distrB}
For any, $(\distrB, \distrU) \in \mathcal{P}_{\eps, \gap}$, if $\distrB$ has probability density function $b$, then: $$b(t) \leq \frac{(1-e^{-\eps})}{e^{-\eps} \gap} \;\; \forall t \in [0, \infty).$$
\end{lemma}

\begin{proof}
Since $b(\cdot)$ is non-negative and integrates to $1$ it must be bounded. Take any $t^* \in \argmax_{t \in [\gap, \infty)} b(t) $. Then,
\begin{align}
 1 & = \int_{0}^{\infty} u(t) \; dt \\
 & = \int_{0}^{t^* - \gap} u(t) \; dt + \int_{t^* - \gap}^{t^*} u(t) \; dt + \int_{t^*}^{\infty} u(t) \; dt \\ 
 & \geq  \int_{0}^{t^* - \gap} e^{-\eps}b(t + \gap) \; dt  + \int_{t^* - \gap}^{t^*} e^{-\eps} b(t^*) \; dt + \int_{t^*}^{\infty} e^{-\eps} b(t) \; dt \\ 
 & = \gap  e^{-\eps} b(t^*) +  \int_{\gap}^{\infty} e^{-\eps}b(t) \; dt\\
 & = \gap  e^{-\eps} b(t^*) + e^{-\eps},
\end{align}
where (3) follows from the indistinguishability definition and (1) and (5) follow since $\distrB$ and $\distrU$ both must integrate to $1$ to be valid probability density functions. Then, $\max\limits_{t \in [0, \infty)} b(t) = b(t^*) \leq \frac{1 - e^{-\eps}}{e^{-\eps} \gap}$.

\end{proof}

\begin{lemma}[$\distrB_i$ determines $\distrU_i$]
\label{lem:determines}
For any $i \in \N$, let $(\distrB_i,\distrU_i) \in \mathcal{P}^{(i)}_{\eps, \gap}$ be Pareto optimal distributions (within $\mathcal{P}^{(i)}_{\eps, \gap}$) with probability density sequences $b_0^{(i)},b_1^{(i)},\ldots$ and $u_0^{(i)},u_1^{(i)},\ldots$ respectively. Then, $\forall k \in \Ints_{> 0}$ it holds that $\uSeq = \max\limits_{j \in [0,i]} e^{-\eps} b^{(i)}_{k+j}$ and $u^{(i)}_0 = \frac{i}{\gap}\left(1 - \sum\limits_{k=1}^{\infty} \frac{\gap}{i} \uSeq \right)$.
\end{lemma}

\begin{proof}
Informally, this proof will argue that if $\distrU_i$ has any ``excess'' probability mass in an interval greater than $0$, we can move that probability mass to the interval at $0$ and reduce the expectation of $\distrU_i$. 
By Lemma~\ref{lem:piecewise}, $\distrB_i$ and $\distrU_i$ are $(\eps, \gap)$-one-sided indistinguishable so $\forall k \in \Ints_{>0}$ it must be that $\uSeq \geq \max\limits_{j \in [0,i]} e^{-\eps} b^{(i)}_{k+j}$. Assume for the sake of contradiction that there is some value $\ell > 0$ for which $u^{(i)}_{\ell} > \max\limits_{j \in [0,i]} e^{-\eps} b^{(i)}_{\ell+j} =: M$. Then, define $\distrU_i'$ to have $u'^{(i)}_{\ell} = M$, $u'^{(i)}_0 = u^{(i)}_{0} + u^{(i)}_{\ell} - M$
and $u'^{(i)}_k = \uSeq$ for all other values of $k$. Then, $\distrU_i'$ is still a valid probability distribution and is $(\eps, \gap)$-indistinguishable from $\distrB_i$, but has lower expected value than $U'$ contradicting the Pareto optimality of $(\distrB_i,\distrU_i)$. The value of $u^{(i)}_0$ follows by requiring that the probability densities integrate to $1$.
\end{proof}

\begin{lemma}[$\{b^{(i)}_k\}$ and $\{u^{(i)}_k\}$ are non-increasing]
\label{lem:non_increasing_step}
For any $i \in \N$, let $(\distrB_i,\distrU_i) \in \mathcal{P}^{(i)}_{\eps, \gap}$ be Pareto optimal distributions (within $\mathcal{P}^{(i)}_{\eps, \gap}$) with probability density sequences $b^{(i)}_0,b^{(i)}_1,\ldots$ and $u^{(i)}_0,u^{(i)}_1,\ldots$ respectively. Then, $\forall k \geq i\;$ it must be that $\bSeq \geq b^{(i)}_{k+1}$ and $\forall k \geq 0$ it must be that $u^{(i)}_{k} \geq u^{(i)}_{k+1}$.
\end{lemma}

\begin{proof}
Suppose that $(\distrB_i, \distrU_i) \in \mathcal{P}^{(i)}_{\eps, \gap}$ are a Pareto optimal pair of distributions with density sequences $\{b^{(i)}_0, b^{(i)}_1,\ldots\}$ and $\{u^{(i)}_0, u^{(i)}_1,\ldots\}$ respectively. We will construct new random variables $(B'_i, U'_i)$ with monotonically non-increasing density sequences $\{b'^{(i)}_0, b'^{(i)}_1,\ldots\}$ and $\{u'^{(i)}_0, u'^{(i)}_1,\ldots\}$ and argue that $\E[\distrB_i'] \leq \E[\distrB]$ and $\E[\distrU_i'] \leq \E[\distrU_i]$. We construct the new density sequences and a permutation $\pi: \N \to \N$ mapping $\{b^{(i)}_k\}$ to $\{b'^{(i)}_k\}$ as follows.

$\Pr[\distrB'_i < \gap] = 0$ by Lemma~\ref{lem:support_\distrB}, so:
\begin{align*}
    &  b'^{(i)}_k = \bSeq = 0, \;\; \forall k \in \Ints, 0 \leq k \leq (i-1)  \\ 
   & \pi(k) = k,  \;\; \forall k \in \Ints, 0 \leq k \leq (i-1). 
\end{align*}

Then, we sort $\{\bSeq\}$ by moving the interval with highest probability mass in $\{\bSeq\}$ (breaking ties to the left) as far to the left as possible in $\{b_k'^{(i)}\}$:
\begin{align*}
    & \forall m \in \Ints, m \geq i:  \\ 
   & \;\; I_m = \argmax_{k \in \N \setminus \{\pi(j) | j < m\}} \bSeq\\ 
   & \;\; \pi(m) = \min_{n \in I_m} n\\ 
   & \;\; b'^{(i)}_{m} = b^{(i)}_{\pi(m)}.
\end{align*}

Finally, by Lemma~\ref{lem:determines}, $\{u'_k\}$ must be determined by $\{b'_k\}$ in order to be Pareto optimal, so take:
\begin{align*}
    u'^{(i)}_k = \begin{cases}
     e^{-\eps} b'^{(i)}_{k} & k \geq i  \\ 
       e^{-\eps} b'^{(i)}_{i} & 1 \leq k \leq (i-1) \\ 
     \frac{i}{\gap}(1-e^{-\eps}) - (i-1) e^{-\eps} b'^{(i)}_{i} & k = 0 
    \end{cases}
\end{align*}

First, we argue that $(\distrB_i', \distrU_i') \in \mathcal{P}^{(i)}_{\eps, \gap}$. $\{b_k'\}$ defines a valid probability distribution since $\{b'^{(i)}_k\}$ is a permutation of $\{\bSeq\}$ so the distribution integrates to $1$. Then, by construction, $\{u'^{(i)}_k\}$ is also a valid probability density sequence and integrates to $1$. By Lemma~\ref{lem:bound_\distrB}, $b^{(i)}_{\pi(i)} \leq \frac{1 - e^{-\eps}}{e^{-\eps} \gap }$ so $u'^{(i)}_0 \geq e^{-\eps} b^{(i)}_{\pi(i)} = e^{-\eps} b'^{(i)}_i$. Hence, the two distributions satisfy the $(\gap, \eps)$-indistinguishability constraint by construction since $b'$ is non-increasing above interval $i$, and $u'^{(i)}_k \geq e^{-\eps} b'^{(i)}_i \; \forall k \leq i$ and $u'^{(i)}_k = e^{-\eps} b'^{(i)}_k \; \forall k \geq i$.

Now, we argue that $\E[\distrB_i'] \leq \E[\distrB_i]$ since $\{b'^{(i)}_k\}$ is a permutation of $\{\bSeq\}$ that shifts probability mass to the left. By construction $\forall t \in [0, \infty)$, it holds that $\Pr[\distrB_i' \leq t] \geq \Pr[\distrB_i \leq t]$. So, 
$$ \E[\distrB_i'] = \int_{0}^{\infty} 1 - \Pr[\distrB_i' \leq t] \; dt \leq \int_{0}^{\infty} 1- \Pr[\distrB_i \leq t] \; dt = \E[\distrB_i].$$

Finally, we want to show that $\E[\distrU_i'] \leq \E[\distrU_i]$. We will analyze the contribution to the expectation coming from intervals below $i$ and above $i$  separately.

Note that the expectation $\E[\distrU_i] = \sum_{k=0}^{\infty} \left(u_k^{(i)} \frac{\gap}{i} \right) \left(\frac{2k+1}{2} \frac{\gap}{i}\right)$ so we can split the difference between the expectations as follows:
\begin{align*}
        2\left(\tfrac{i}{g}\right)^2 (\E[\distrU_i] - \E[\distrU_i']) &= \sum_{k=i}^{\infty}  (2k+1) (\uSeq - u'^{(i)}_k) + \sum_{k=0}^{i-1}  (2k+1) (\uSeq - u'^{(i)}_k).
\end{align*}
Now, we state the following two observations, which we will apply repeatedly in the remainder of the proof:
\begin{enumerate}[label=(\roman*)]
    \item$\forall k \geq i: \; u^{(i)}_{\pi(k)} \geq e^{-\eps} b^{(i)}_{\pi(k)} = u'^{(i)}_k$, by indistinguishability of $\distrB$ and $\distrU$ and the definition of $\distrU'$.
    \item $\pi(\cdot)$ is a bijection on $[i, \infty)$ so $\sum_{k=i}^{\infty}  u^{(i)}_k = \sum_{k=i}^{\infty} u^{(i)}_{\pi(k)}$.
\end{enumerate}
By properties (i) and (ii) above, there is ``excess probability density'' above interval $i$ in $\distrU$ compared to $\distrU'$ of 
\begin{align*}
    M = \sum_{k = i}^{\infty} \uSeq - u'^{(i)}_k =  \sum_{k = i}^{\infty} u^{(i)}_{\pi(k)} - u'^{(i)}_k \geq 0.
\end{align*}
Since $\sum_{0}^\infty u'^{(i)}_k = \sum_{0}^\infty \uSeq$, by symmetry there is excess probability mass of $M$ below $i$ in $U'$ compared to $U$:
\begin{align*} M = \sum_{k = 0}^{i-1} u'^{(i)}_k - \uSeq. \end{align*}

Since $\{u'^{(i)}_k\}_{k=1}^{\infty}$ is non-increasing and by properties (i) and (ii) above the $\{u'^{(i)}_k\}$ are a permutation of $\{u^{(i)}_k\}$ with some values increased, the difference in expectations between $\distrU$ and $\distrU'$ above interval $i$ is minimized by putting all of the excess probability mass $M$ in interval $i$, so:
\begin{align*}
 \sum_{k=i}^{\infty} (2k+1) (\uSeq - u'^{(i)}_k)   \geq M(2i+1).
\end{align*}

To analyze the difference in expectations coming from intervals in $k \in [0, i-1]$, we first argue that $\distrU'$ puts more probability mass on $0$ than $\distrU$, that is $u'^{(i)}_0 \geq u^{(i)}_0$
In particular, we will argue that $\sum_{k=1}^{\infty} u'^{(i)}_k \leq \sum_{k=1}^{\infty} \uSeq$. By indistinguishability, $\forall k \geq 1 \; \uSeq \geq e^{-\eps} b^{(i)}_{k+i}$ and $\uSeq \geq e^{-\eps} b^{(i)}_{k}$ so \begin{align*}\sum_{k=1}^{\infty} \uSeq & \geq \sum_{k=1}^{\pi(i)-i-1} e^{-\eps} b^{(i)}_{k+i} + \sum_{k=\pi(i)-i}^{\pi(i)-1} e^{-\eps} b^{(i)}_{i} + \sum_{k=\pi(i)}^\infty e^{-\eps} \bSeq \\ 
& = \sum_{k=i}^{\pi(i)-1} e^{-\eps} \bSeq + \sum_{k=\pi(i)}^{\infty} e^{-\eps} \bSeq + (i-1) e^{-\eps} \distrB_i = \sum_{k=1}^{\infty} u'^{(i)}_k. \end{align*}

Next, we argue that $u^{(i)}_k$ is non-decreasing on $[1, i-1]$. By Lemma~\ref{lem:determines}, for $\forall k \in \Ints, (i-1) \geq k \geq 1: \; \uSeq = \max\limits_{j \in [i, i+k]} e^{-\eps} b^{(i)}_j \leq \max\limits_{j \in [i, i+k+1]} e^{-\eps} b^{(i)}_j = u^{(i)}_{k+1}$. Therefore, putting the excess probability mass $M$ in $U'$ compared to $U$ as far to the right as possible gives \begin{align*}\sum_{k=0}^{i-1} (2k+1) \uSeq \geq u'^{(i)}_0 + \sum_{k=0}^{i-1} \left( (2k+1) \left(u'^{(i)}_k - \frac{M}{i-1}\right) \right) = \left(\sum_{k=0}^{i-1} (2k+1) u'^{(i)}_k\right) - M(i+1),\end{align*}
so \begin{align*}\sum_{k=0}^{i-1} (2k+1) (u'^{(i)}_k - \uSeq) \leq M(i+1).\end{align*}
Thus, we conclude that
\begin{align*}
        2\left(\tfrac{i}{g}\right)^2 (\E[\distrU] - \E[\distrU']) &= \sum_{k=i}^{\infty}  (2k+1) (\uSeq - u'^{(i)}_k) + \sum_{k=0}^{i-1}  (2k+1) (\uSeq - u'^{(i)}_k) \\ 
         & \geq M(2i+1) - M(i+1)\\
         & \geq Mi \geq 0,\\ 
\end{align*}
giving $\E[\distrU] \geq \E[\distrU']$.
\end{proof} 

\subsubsection*{Pareto Frontier of $\distrB_i, \distrU_i$:}
\label{sec:pareto_distributions}

Now, we use the properties of Pareto optimal $\distrB_i, \distrU_i$ to give an exact characterization of the probability density functions of Pareto optimal $\distrB_i, \distrU_i$:

\begin{lemma}
\label{lem:exact_distribution}
For any $i \in \N$, let $S_L = \{(\distrB_i, \distrU_i) \in \mathcal{P}^{(i)}_{\eps, \gap} \; : \; b^{(i)}_i = L \text { and } (\distrB_i, \distrU_i) \text{ are Pareto optimal} \}$ be all distributions in the Pareto frontier of $\mathcal{P}^{(i)}_{\eps, \gap}$ where $b^{(i)}_i$ is fixed to be some value $L \leq \frac{1 - e^{-\eps}}{e^{-\eps} \gap}$. Then, either $S_L = \emptyset$ or $S_L$ contains a single pair of distributions where letting $n = \lfloor \frac{i}{\gap} \cdot \frac{1}{L}  \rfloor$: 
\begin{enumerate}[label=(\roman*)]
    \item $\bSeq = L$ for $k \in [i, n]$, $b^{(i)}_{n+1} = \frac{i}{\gap} (1 - \frac{\gap}{i} n L)$, and $\bSeq = 0$ for all other values of $k$.
    \item $\uSeq = e^{-\eps} L$ for $k \in [1,n]$, $u^{(i)}_{n+1} = e^{-\eps} b^{(i)}_{n+1}$, $u^{(i)}_0 = \frac{i}{\gap} ( 1 - e^{-\eps} - \tfrac{(i-1)}{i} \gap e^{-\eps} L)$ and $\uSeq = 0$ for all other $k$.
\end{enumerate}
so $\distrB_i$ is a ``nearly uniform'' distribution above $\gap$ with any excess probability mass in the final constant interval, and $\distrU_i$ has the same probability mass as $\distrB_i$ discounted by $e^{-\eps}$ except in a small band around $0$ where it may have inflated probability mass. 
\end{lemma}

\begin{proof}
First, note that fixing $b^{(i)}_i = L$, by Lemma~\ref{lem:determines} we have that $\uSeq$ is fully determined by $L$ for $k \in [0,i)$. Therefore, for any Pareto optimal $\distrU_i, \distrB_i$ with $b^{(i)}_i = L$: 
\begin{align*}\E[\distrU_i] = \Pr[\distrU_i < \gap]\E[\distrU_i | \distrU_i < \gap] + \Pr[\distrU_i \geq \gap]\E[\distrU_i | \distrU_i \geq \gap] = C_L + (1-e^{-\eps}) \E[\distrB_i],\end{align*}
where $C_L$ is a constant determined by $L$. Therefore, there is a unique minimizer of $\E[\distrU_i]$ and $\E[\distrB_i]$ over $S_L$ that is obtained by minimizing $\E[\distrB_i]$. Since $\distrB_i$ is monotonically non-increasing above $i$, the distribution that minimizes its expectation puts mass equal to $b^{(i)}_i = L$ at as many intervals as possible giving $n = \lfloor \frac{i}{\gap} \frac{1}{L} \rfloor$ intervals with $\bSeq = L$ and any remaining mass needed to make the distribution integrate to $1$ in the final interval, yielding the unique optimal distributions for $(\distrB_i, \distrU_i)$.
\end{proof}

Taking limits as $i \to \infty$ of each distribution in the set of distributions from Lemma~\ref{lem:exact_distribution} yields exactly the set of zero-inflated Uniform distributions in Theorem~\ref{thm:unif_opt}, so we conclude that any optimizer of a weighted sum objective must come from this set of distributions and hence the Pareto frontier consists of Zero-inflated Uniform distributions.

\subsubsection{Optimal choice of parameter $\pUnif$}
\label{sec:param_choice}

Finally, we derive the optimal choice of parameter $\pUnif$ given $\eps, \gap$ and weighting parameter $\weightB$. From Theorem~\ref{thm:privacy_distributions} the zero-inflated Uniform with parameters $\eps, \gap$ has expectation:
$\E[\distrB] = \frac{1}{2} \gap \left(\pUnif + \frac{\pUnif}{\pUnif - e^{-\eps}} \right)$ and $\E[\distrU] = \frac{1}{2} \gap \left( \frac{\pUnif^2}{\pUnif - e^{-\eps}} \right)$. Therefore, by Pareto optimality of the zero-inflated Uniform proven in Section~\ref{sec:pareto_proof_sec}, for any $\weightB \in [0,1]$, the weighted sum of the expectations can be optimized by choosing \begin{align*}
    \pUnif^* \in \argmin_{\pUnif \in (e^{-\eps}, 1]} \weightB \left(\pUnif + \frac{\pUnif}{\pUnif - e^{-\eps}}\right) + (1-\weightB)  \frac{\pUnif^2}{\pUnif - e^{-\eps}}. 
\end{align*} 

This objective is convex on $(e^{-\eps}, 1]$ as it has second derivative with respect to $\pUnif$ of $ \weightB \left(1 + \frac{1}{\pUnif-e^{-\eps}}\right) + (1-\weightB) \frac{\pUnif^2}{\pUnif - e^{-\eps}} > 0$ for any $\weightB \in [0,1]$ and $\pUnif \in (e^{-\eps}, 1]$.

The first derivative of this objective with respect to $\pUnif$ is $ \frac{1}{(\pUnif - e^{-\eps})^2} \left(\weightB (e^{-2\eps} - e^{-\eps}) + \pUnif(\pUnif - 2e^{-\eps})\right)$. Note that for any $\weightB$, the derivative begins at a negative value on the interval $(e^{-\eps}, 1]$ and is increasing on this interval. Therefore, letting $\hat{\pUnif}$ denote the value at which the first derivative is $0$, we obtain $\hat{\pUnif} = e^{-\eps} \left(1 + \sqrt{1 + e^{\eps} \frac{\weightB}{1 - \weightB}} \right)$. Since $\hat{\pUnif}$ must fall in the interval $(e^{-\eps},1]$ we take $\pUnif^* = \min\{1, \hat{\pUnif}\}$ to get the optimal $\pUnif$ given in Algorithm~\ref{alg:unif_mechanism}, where $\pUnif^*$ is optimal since the utility function must be decreasing on $[e^{-\eps}, 1)$ in the case that $\hat{\pUnif} > 1$.

\subsection{Proof of Theorem~\ref{thm:impossibility} (Impossibility of Two-Sided DP)}
\label{appendix:impossibility}

We prove the result for each of following definitions of ``neighboring'' separately:
\begin{enumerate}[label=(\arabic*)]
    \item \emph{Add or remove a batched comment.} Consider any input $\Arriv$ where an instance of batching occurs at some time $\ts$. Let $\Arriv'$ be identical to $\Arriv$, except some comment $\act$ that arrived in a batch at time $\ts$ does not arrive at all in $\Arriv'$. Then on input $\Arriv'$, since any valid comment posting mechanism cannot generate fake data, for any $\delay > 0$ and time $\finiteTime = \ts + \delay$, the mechanism outputs $\act$ at time $\finiteTime$ with probability $0$. However, if the mechanism is $(\eps, \delta)$-DP with $\eps < \infty$, then for any release time $\finiteTime$ the mechanism outputs $\act$ within time $\finiteTime$ with probability at most $\delta < 1$ and so the mechanism violates the eventual release of all comments property. 
    \item \emph{Move a batched comment to another arrival time where it is no longer batched.} Consider any input $\Arriv$ with an instance of batching that occurs at some time $\ts$. Fix any time horizon $\finiteTime = \ts + \delay$ where $\delay > 0$. Define $\Arriv'$ to be an identical set with one comment $\act$ moved from time $\ts$ to time $\finiteTime$. Since a valid comment posting mechanism must delay comments and cannot generate fake data, the mechanism outputs comment $\act$ at time $\finiteTime$ or later on input $\Arriv'$ with probability $1$. However, if the mechanism is $(\eps, \delta)$-DP with $\eps < \infty$ then it must delay comment $\act$ until at least time $\finiteTime = \ts + \delay$ with probability at least $1-\delta$. Taking $\delay$ to be arbitrarily large, the mechanism violates the eventual release of all comments property for any $\delta < 1$. 
    \item \emph{Move a batched comment by at most $\gap$ units of time to another arrival time where it is no longer batched.} Let $\Arriv^{(1)}$ be an input where a single comment arrives every $\gap$ units of time. Then, define $\Arriv^{(1)'}$ to be a neighboring input to $\Arriv^{(1)}$ where $\act_2$ arrives in a batch with $\act_1$ at time $0$. Define $\Arriv^{(2)}$ to be a neighboring input to $\Arriv^{(1)'}$ where $\act_1$ and $\act_2$ arrive separately with $\act_1$ at time $\gap$ and $\act_2$ at time $0$ and so on:
\begin{align*}
 \Arriv^{(1)} & = \{\act_1, \ts = 0\},  \{\act_2, \ts = \gap\}, \{\act_3, \ts = 2\gap\},\ldots   \\ 
 \Arriv^{(1)'} & = \{\act_1, \act_2, \ts = 0\}, \emptyset, \{\act_3, \ts = 2\gap\},\ldots \\ 
 \Arriv^{(2)} & = \{\act_2, \ts = 0\}, \{\act_1, \ts = \gap\}, \{\act_3, \ts = 2\gap\},\ldots \\ 
 \Arriv^{(2)'} &  = \{\act_2, \ts = 0\}, \{\act_1, \act_3, \ts = \gap\}, \emptyset,\ldots
\end{align*}

Now, for any $j$: $\Arriv^{(j)}$ and $\Arriv^{(j)'}$ are neighbors and $\Arriv^{(j)}$ and $\Arriv^{(j-1)'}$ are neighbors. On input $\Arriv^{(j)}$, comment $\act_1$ arrives at time $j \gap$ and so any valid comment posting therefore posts $\act_1$ at time $j \gap$ or later with probability $1$ since it can only delay comments. Likewise, because $\Arriv^{(j-1)'}$ neighbors $\Arriv^{(j)}$ and the mechanism cannot generate fake data, any $(\eps,\delta)$-DP mechanism releases $\act_1$ at a time earlier than $j \gap$ with probability at most $\delta$ on input $\Arriv^{(j-1)'}$. Since $\Arriv^{(j-1)}$ neighbors $\Arriv^{(j-1)'}$, the mechanism releases $\act_1$ at a time earlier than $j$ with probability at most $2\delta$ on this input. Thus, on input $\Arriv^{(1)}$, comment $\act_1$ gets posted before time $j \gap$ with probability less than $2j\delta$. This suggests that the comment gets delayed by at least $\maxDelay$ with probability at least $1-2\delta(\tfrac{\maxDelay}{\gap} +1)$.
\end{enumerate}

\subsection{Proof of Proposition~\ref{prop:percentile_osdp} (Hypothesis Testing Interpretation of OSDP)}
\label{app:percentile}

Fix comment $\act_1$ and let $\act_2$ denote the closest comment to arrive in $\anonSet$. Let $R$ denote the rejection region of the adversary's chosen hypothesis test. Let $\inpBatch$ be any arrival set where $\ts_1 = \ts_2$. Let $\inp_d$ be an identical arrival set, except that $\act_2$ arrives unbatched $d$ units of time after $\act_1$ (so $\ts_2 - \ts_1 = d$) and let $\inp_d'$ be an identical arrival set except that $\act_1$ arrives $d$ units of time after $\act_2$. Then, conditioning on the event that $\ts_2 - \ts_1 \leq g$, we have that for any rejection region $R$:
$$\textit{\textit{Type I Error}} \geq \sum_{n = 0}^{\gap}\Pr[\ts_2 -\ts_1 = d ; \distr]   \Pr[\mech(\inp_d) \in R] + \sum_{n = 0}^{\gap}\Pr[\ts_1 -\ts_2 = d ; \distr] \cdot  \Pr[\mech(\inp_d') \in R]$$
$$\geq e^{-\eps}  \perc(\gap)  \Pr[\mech(\inpBatch) \in R] = e^{-\eps}  \perc(\gap)  \textit{Power},$$
where the second line follows from the one-sided differential privacy guarantee on $\gap$-adjacent inputs.

\section{Estimation of Batching Deanonymization Risk Statistics}
\label{app:deanon_attack}

Recall that in Section~\ref{sec:intro}, we provided statistics on the rate of batching at a peer-reviewed conference. We used these statistics in Figure~\ref{fig:posteriors} to estimate the linkage risk arising due to observing batched comments. In this section, we provide details about the measurement method used to estimate the batching statistics.

In order to estimate the prevalence of batching in the peer-review process of a conference, we measure the following statistics. For any individual reviewer or meta-reviewer, we order all of their comments on all papers in increasing order of post time. If two comments arrive immediately next to each other in this sequence and were made on different papers, we consider these to be ``consecutive comments from the same (meta)-reviewers on different papers.'' Note that this excludes comments that are made on the same paper by the same (meta)-reviewer consecutively, because consecutive comments by the same (meta)-reviewer on the same paper do not generate additional linkage risk for the (meta)-reviewer. For example, consider the following sequence of comment arrivals from a single (meta)-reviewer (where units of time are minutes from the start of the commenting period):
\begin{align*}
    (\act_1, \paper_1, \ts_1 = 0), (\act_2, \paper_2, \ts_2 = 5),  (\act_3, \paper_2, \ts_2 = 6), (\act_4, \paper_2, \ts_2 = 8) (\act_5, \paper_3, \ts_3 = 100).
\end{align*}
In this example, we count the first two comments ($\act_1$ and $\act_2$) and the last two comments ($\act_4$ and $\act_5$) as consecutive arrivals on different papers. We then capture the rate of batching under $5$ minutes by computing the number of consecutive comments that arrive within $5$ minutes of each other divided by the total number of consecutive comment arrivals. So, in the example above, the rate of batching is $50\%$ since comments $\act_1$ and $\act_2$ arrive within $5$ minutes of one another, while $\act_4$ and $\act_5$ do not. Applying this measurement method to a dataset of comments made by reviewers and meta-reviewers on papers at a top Computer Science conference, we find that there is a $30.10\%$ chance that a comment arrives in a batch with a consecutive comment from the same (meta)-reviewer. 

For a baseline, we additionally compute how often comments from \emph{different} (meta)-reviewers may appear at times close to each other. We look at each pair of distinct reviewers from the set of all reviewers. We then calculate whether any pair of comments from these two (meta)-reviewers arrived within a cutoff of $5$ minutes of one another. We find that there is a $0.66\%$ chance that a randomly chosen pair of (meta)-reviewers makes a pairs of comments that arrive within $5$ minutes of one another. We note that the first statistic capturing the rate of batching excludes reviewers who made only a single comment in the entire conference, as it is not possible for these reviewers to engage in batching. In contrast, the second statistic capturing the baseline rate of close arrivals includes cases where a reviewer makes only a single comment. These comments are counted in the statistic, since any comment may appear to be batched with an anonymized comment made by a different reviewer from the perspective of an observer who does not know reviewer identities.

\section{A Queue-Based Mechanism for Privacy Against Batched Timing Attacks}
\label{app:queue_mech}

In this section, we discuss an alternative privacy formulation that we call ``$\eps$-batching privacy'' and give an algorithm that satisfies privacy under this formulation by delaying comments using a queue to preserve privacy. In doing so, our queue-based mechanism \emph{preserves the ordering in which comments arrive}, a property that may be useful in certain applications. The privacy guarantees are not directly comparable to $(\eps, \gap)$-OSDP because we make substantially different sets of assumptions in the adversarial model. However, one can think of both approaches as responses to the impossibility results for standard two-sided proven in Section~\ref{sec:dp_impossibility}. While $(\eps, \gap)$-OSDP relaxes two-sided DP by introducing a bound $\gap$ on the gap between unbatched comments and by making the notion of neighbors asymmetric, $\eps$-batching privacy introduces distributional assumptions on the inputs that capture an adversary's uncertainty about comment arrivals.

\subsection{Problem Formulation}

In this problem formulation, we assume that comment arrivals are drawn i.i.d. from some unknown distribution over papers and reviewers. We call this the arrival process. We assume \emph{discrete time} comment arrivals over an infinite time horizon so comments arrive at each time-step drawn from this unknown distribution. 

First, we present the arrival process if no batching occurs. In the absence of batching, a single comment arrives at every unit of time. We make an \emph{i.i.d. assumption} on arrivals. At each time-step, the paper-reviewer pair associated with the comment is drawn independently from a (potentially unknown) probability distribution $\distr$ over $\Papers \times \Reviewers$ (where $\Papers$ is the set of all papers and $\Reviewers$ is the set of all reviewers). For instance, $\distr$ could be a uniform distribution over $\Papers \times \Reviewers$ although it need not be uniform or even known to the algorithm. We say $\Arriv \leftarrow \procNoBatch$ if the arrivals are drawn from this no-batching process.

An instance of \emph{potential batching} consists of multiple comments. The batch arrives at a single time-step, but the adversary is uncertain as to which papers and reviewers are in the batch. Thus, when potential batching occurs, the arrival process remains the same except for one modification--- batches consisting of more than one comment arrive at specific fixed time-steps. Formally, let $\batchTimeSet$ be a multi-set of time-steps at which batching occurs. The arrival process proceeds as follows:
\begin{itemize}
    \item On time-steps not contained in $\batchTimeSet$, no batching occurs and a single comment arrives.
    \item For each time-step contained in $\batchTimeSet$, an additional comment arrives due to batching. For instance, if $\batchTimeSet = \{10,10,15\}$ then a single comment arrives at each time-step, but two additional comments arrive at time $10$ due to batching and one additional comment arrives at time $15$ due to batching.
\end{itemize} The paper-reviewer pairs associated with the batched comments are drawn independently with replacement from distribution $\distr$. We say that $\Arriv \leftarrow \procBatchMulti$ if the arrivals are drawn from this process with batchings occurring at time-steps in $\procBatchMulti$. We allow comments to arrive according to $\procBatchMulti$ for any finite multi-set of time-steps $\batchTimeSet$. We do not assume any prior knowledge of either $\batchTimeSet$ nor $|\batchTimeSet|$.

Then, we define a comment posting mechanism to be $\eps$-batching private in this formulation, if the mechanism obscures whether the inputted comment arrival set arrived per the batching process (with any number of batches) or the no batching process (whereby $0$ batches appeared):

\begin{definition}[Batching Privacy]
A comment posting mechanism $\mech$ is \emph{$\eps$-batching private} with respect to arrival processes $(\procNoBatch, \procBatchMulti)$ if for all time horizons $\finiteTime \geq 1$, all finite batching multi-sets $\batchTimeSet$, and any output of the mechanism between time $1$ and $\finiteTime$, $\out_{\finiteTime}$:
\begin{align*}
   & \Pr[\mech(\Arriv) = \out_{\finiteTime} ; \Arriv \leftarrow \procBatchMulti] \leq e^\eps \Pr[\mech(\Arriv) = \out_{\finiteTime} ; \Arriv \leftarrow \procNoBatch] \: \text{and}\\
    & \Pr[\mech(\Arriv) = \out_{\finiteTime} ; \Arriv \leftarrow \procNoBatch] \leq e^\eps \Pr[\mech(\Arriv) = \out_{\finiteTime} ; \Arriv \leftarrow \procBatchMulti].
\end{align*}
\end{definition}

Note that unlike typical differential privacy formulations, this notion of privacy requires distributional assumptions on the data-generating process as we assume that comments are generated by an i.i.d. arrival model.

\subsection{Results}

Under this formulation, we design a mechanism described in Algorithm~\ref{alg:queue_anon_0} that delays comments by deploying them to a queue. The algorithm guarantees \emph{perfect batching privacy} ($\eps = 0$), as shown in the following result. 
\begin{proposition}[Privacy]
Algorithm~\ref{alg:queue_anon_0} guarantees perfect batching privacy $(\eps = 0)$ for comments arriving according to $\procNoBatch$ and $\procBatchMulti$ for any $\batchTimeSet$. 
\end{proposition}

\begin{proof}
Fix a time horizon $\finiteTime$ and multi-set of batching times $\batchTimeSet$. We let $\distr(\act)$ denote the probability of observing the comment $\act$ under distribution $\distr$. When the algorithm is applied to comments drawn according to the no batching process, one comment arrives at each time-step and all comments are posted immediately so by the i.i.d. assumption, $\Pr[\mech(\Arriv) = \act_{1:\finiteTime}; \Arriv \leftarrow \procNoBatch] = \prod_{i=1}^{\finiteTime} \distr(\act_{\ts})$.

If comments were drawn according to the process where batching occurred at times $\batchTimeSet$, then at any time-step before the first instance of batching occurs the mechanism posts the single comment that arrives so the probability of observing output $\{\act\}$ is $\distr(\act)$ independent of other-timesteps. On the first instance of batching, the mechanism posts one of the batched comments chosen uniformly at random from the batch, so due to the i.i.d. arrivals of the batch the probability of observing this output $\{\act\}$ at this time-step is also $\distr(\act)$. At any later time-step, the algorithm posts the comment at the top of the queue, which consists of previous comments that arrived i.i.d. drawn from $\distr$. Therefore, the probability of observing any output is still  $\Pr[\mech(\Arriv) = \act_{1:\finiteTime}; \Arriv \leftarrow \procBatchMulti] = \prod_{i=1}^{\finiteTime} \distr(\act_{\ts})$.
\end{proof}

\begin{algorithm}[t]
   \caption{Queue Mechanism}
   \label{alg:queue_anon_0}
\begin{algorithmic}
   \STATE Initialize empty queue $\queue = \emptyset$
   \FOR{\ts = $1,2,\ldots$}
    \IF{set of batched comments $\Arriv$ arrives}
    \IF{$\queue \not = \emptyset$}
    \STATE Dequeue comment $\act'$ from $\queue$ and post it.
    \STATE Enqueue all comments in $\Arriv$ to $\queue$ in a random order.
    \ELSE
    \STATE Choose $\act \in \Arriv$ uniformly at random to post.
    \STATE Enqueue all comments in $\Arriv \setminus \{\act\}$ to $\queue$ in a random order.
    \STATE Post comment $\act$ immediately.
    \ENDIF
    \ELSIF{a single comment $\act$ arrives}
    \IF{\queue $\not= \emptyset$}
    \STATE Dequeue comment $\act'$ from $\queue$ and post it.
    \STATE Enqueue comment $\act$ to $\queue$.
    \ELSE
    \STATE Post comment $\act$.
    \ENDIF
\ENDIF
\ENDFOR
\end{algorithmic}
\end{algorithm}

The algorithm delays comments by a deterministic value depending on the number of batched comments that have arrived already.
\begin{proposition}[Delay]
If comments arrive according to $\procBatchMulti$, then Algorithm~\ref{alg:queue_anon_0} adds worst-case delay to any comment equal to $|\batchTimeSet|$. 
\end{proposition}

\begin{proof}
After the last instance of batching in $\batchTimeSet$, there are $\finiteTime + |\batchTimeSet|$ comments that have arrived in total. The mechanism posts the earliest-arriving comment at each time-step and delays the incoming comment so the queue has length $|\batchTimeSet|$ and any single incoming comment is delayed for $|\batchTimeSet|$ timesteps before being posted. Any comments arriving before all instances have batching have occurred are delayed by the number of additional comments arriving due to batching at an earlier time-step, so have delay less than $|\batchTimeSet|$.  
\end{proof}

In fact, this perfectly private mechanism is optimal for this privacy formulation as it achieves the best possible worst-case delay to any comment at any value of $\eps$. In particular, at any setting of $\eps$ any batching-private comment posting mechanism must delay a comment by at least $|\batchTimeSet|$ in the worst-case:
\begin{proposition}[Lower Bound, Minimum Delay]
\label{claim:lower_bound_multi_batch}
Any comment posting mechanism guaranteeing $\epsilon$-batching privacy with any $\epsilon < \infty$ for comments arriving according to $\procNoBatch$ and $\procBatchMulti$ must introduce delay of at least $|\batchTimeSet|$ to at least one comment when applied to comments arriving according to $\procBatchMulti$.
\end{proposition}

It follows immediately that since the Queue Mechanism (Algorithm~\ref{alg:queue_anon_0}) achieves this lower bound it is optimal among $\eps$-batching private mechanisms in minimizing worst-case delay:

\begin{corollary}
For any setting of privacy parameter $\eps$, Algorithm~\ref{alg:queue_anon_0} is optimal among $\eps$-batching private comment posting mechanisms in minimizing the worst-case delay added to any comment.
\end{corollary}

\begin{proof}
Let $\finiteTime' = \max\{\batchTimeSet \}$ be the latest time-step when batching occurs and $\finiteTime = \finiteTime' + |\batchTimeSet|$. Then, if comments arrive according to $\procBatchMulti$, $\finiteTime' + |\batchTimeSet|+1 = \finiteTime + 1$ comments arrive up until time $\finiteTime'+1$. Assume for the sake of contradiction that all of the comments arriving before time $\finiteTime'+1$ are posted with delay strictly less than $|\batchTimeSet|$. Then, when acting on comments arriving according to $\procBatchMulti$, the mechanism must post at least $\finiteTime+1$ comments within time horizon $\finiteTime$ (with probability $1$). However, under arrival process $\procNoBatch$, only $\finiteTime$ comments have arrived up until $\finiteTime$, so no mechanism can ever output $\finiteTime+1$ comments up until time $\finiteTime$. Hence, any output of the mechanism up until time $\finiteTime$ on comments arriving per $\procBatchMulti$ contains $\finiteTime+1$ comments with probability $1$, while for comments arriving per $\procNoBatch$ any output up until time $\finiteTime$ contains $\finiteTime+1$ comments with probability $0$. 
\end{proof}

The above formulation and corresponding queue-based mechanism offer an alternative approach to provide privacy in light of the impossibility results for two-sided DP. Here, we relax the problem by introducing distributional assumptions on inputs to the mechanism. While this does not yield a privacy-delay trade-off in $\eps$, it allows for a mechanism that preserves the ordering of comments. As noted in Section~\ref{sec:discussion}, an interesting direction of future work is to understand how we might make the Zero-Inflated Uniform Mechanism order-preserving as well.

\section{Additional Experimental Results}
\label{app:experiments}

In this section, we provide experimental results that augment those presented in the main text. 

\subsection{Wikipedia}
\label{appendix:experiments_wiki}

In the main text, we showed results setting $\gap = 11$ minutes by choosing the 25th percentile of prior inter-arrival times for the category ``21-st century American Politicains'' and $\gap = 36$ minutes at the 50th percentile. Here, we provide additional results, setting $\gap = 79$ minutes based on the 75th percentile of the inter-arrival distribution as shown in Figure~\ref{fig:wiki_p75} and Table~\ref{fig:delay_wiki_75}. Algorithm~\ref{alg:unif_mechanism} adds significantly higher delay at this setting of $\gap$, and consequently the adversary's batched timing linkage attack performs quite poorly. For instance, taking $\eps = 0.5$
corresponds to an average delay of roughly $3$ hours and maximum delay of $6$ hours, but renders the attack highly inaccurate: the attack now achieves around $80\%$ recall at $10\%$ precision compared to the non-private baseline which achieves $85\%$ recall at $80\%$ precision. 

\begin{figure}[h]
\centering
  \includegraphics[width=.5\linewidth]{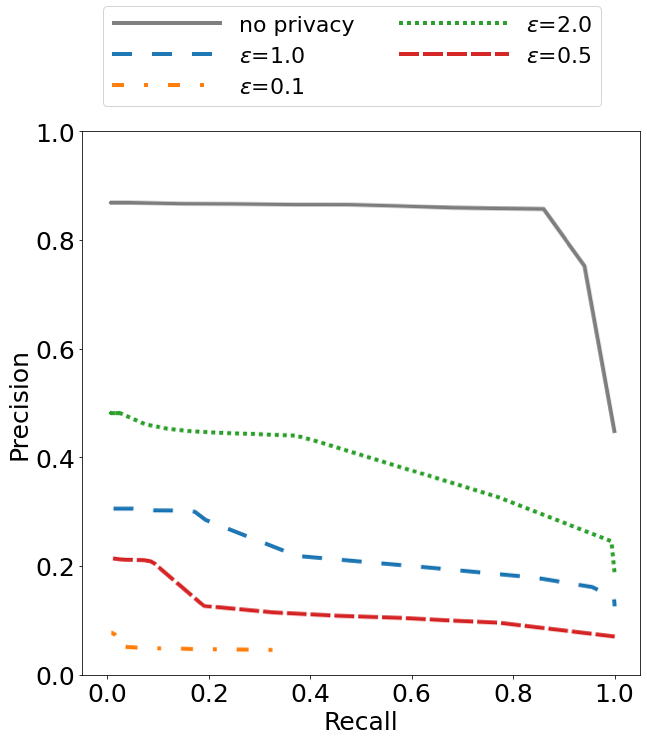}
  \caption{Accuracy in linking pairs of Wikipedia article revisions within the category “21st-century American
Politicians” based on batched timing (averaged over 5 runs of the randomized privacy mechanism) for $\gap$ set to 79 minutes.}
  \label{fig:wiki_p75}
\end{figure}

\begin{table}[ht]
    \centering
    \begin{tabular}{|l|cccc|cccc|}
    \hline
         & \multicolumn{4}{c|}{Mean Delay} & \multicolumn{4}{c|}{Maximum Delay} \\\hline
    & $\eps = 0.1$ & $\eps = 0.5$ & $\eps = 1.0$ & $\eps=2.0$ &  $\eps = 0.1$ & $\eps = 0.5$ & $\eps = 1.0$ & $\eps=2.0$ \\ \hline
    $\: \gap = 79$ & 820 & 192 & 115 & 77 & 1615 & 360 & 205 & 129 
 \\\hline
    \end{tabular}
    \caption{Mean and maximum delay (in minutes) added to Wikipedia article revisions within the category ``21st-century American Politicians'' for $\gap$ set to the $75$th percentile of the historical inter-arrival distribution.}
    \label{fig:delay_wiki_75}
\end{table}

\subsection{Bitcoin}
\label{appendix:experiments_btc}

In the main text, we showed results using Algorithm~\ref{alg:unif_mechanism} with $\gap$ set to the median of the historical inter-arrival times of transactions sent to a given output address (with a default of $10$ minutes when there were no prior transactions.) In this section, we give results for alternative settings of $\gap$. In Figure~\ref{fig:btc_success_25} and Figure~\ref{fig:btc_delay_25} we show the delay added to comments and the success of attacks when $\gap$ is set to a more lenient value based on the $25$-th percentile of historical transaction inter-arrival times. In Figure~\ref{fig:btc_success_75} and Figure~\ref{fig:btc_delay_75} we show results for a stricter setting of $\gap$ to the $75$-th percentile of historical transaction inter-arrival times. 

\begin{figure}[ht]
    \centering
    \includegraphics[width=0.8\linewidth]{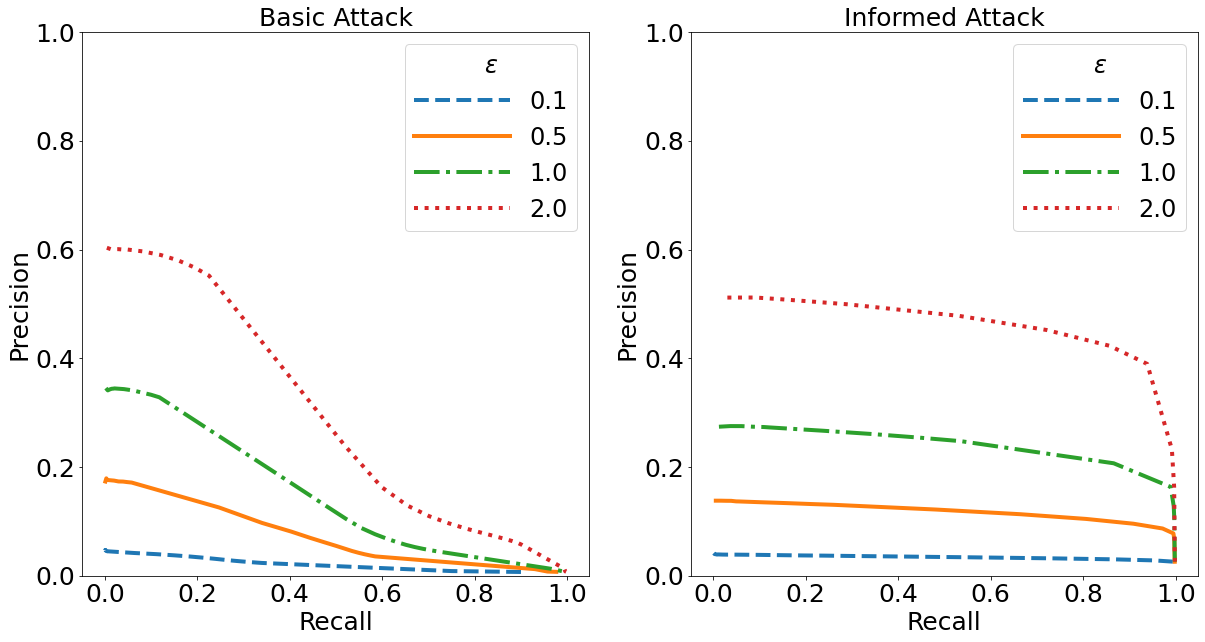}
    \caption{Performance of basic and informed attacks on Bitcoin transactions when $\gap$ is set to the 25th percentile historical inter-arrival time for an output address.}
    \label{fig:btc_success_25}
\end{figure}

\begin{figure}[ht]
    \centering
    \includegraphics[width=0.7\linewidth]{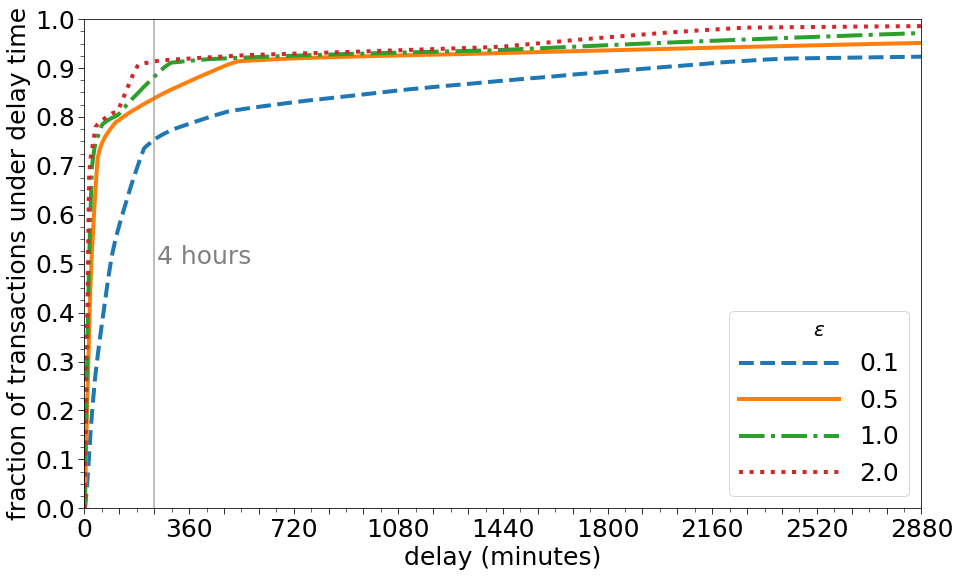}
    \caption{Cumulative distribution of delay added to batched Bitcoin transactions (averaged over 5 trials). Delay is drawn from a privacy-preserving uniform distribution with $\gap$ set to the 25th percentile of the inter-arrival time of transactions to an output address within the past $7$ days.}
    \label{fig:btc_delay_25}
\end{figure}

\begin{figure}[ht]
    \centering
    \includegraphics[width=0.8\linewidth]{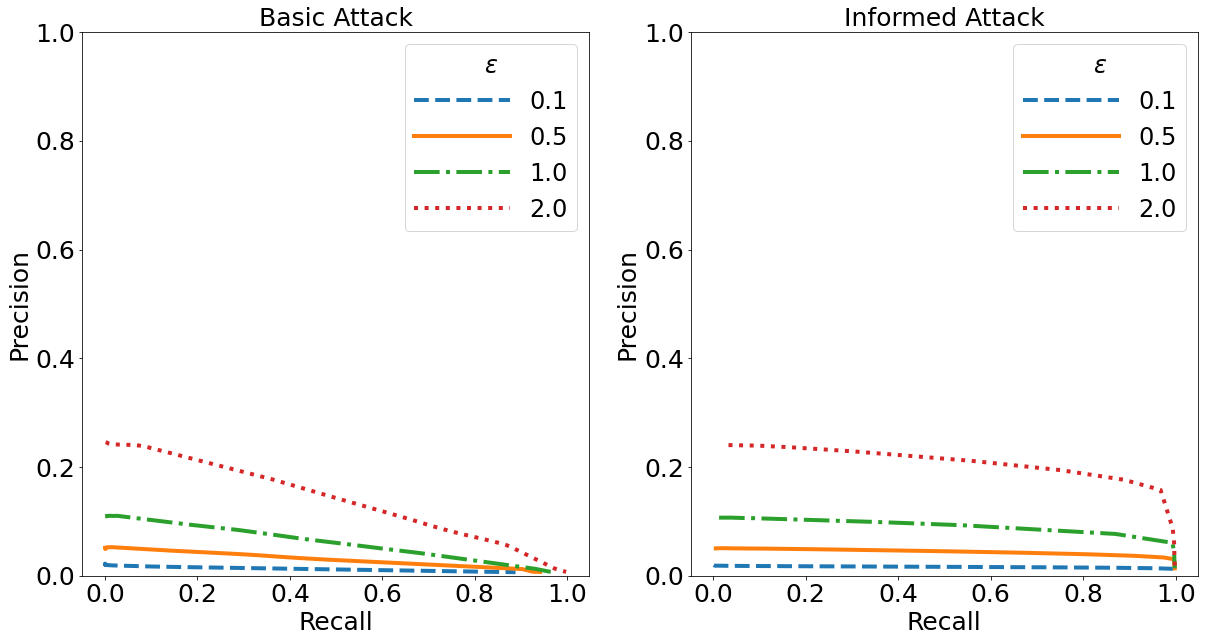}
    \caption{Performance of basic and informed attacks on Bitcoin transactions when $\gap$ is set to the median historical inter-arrival time for an output address.}
    \label{fig:btc_success_75}
\end{figure}

\begin{figure}[ht]
    \centering
    \includegraphics[width=0.7\linewidth]{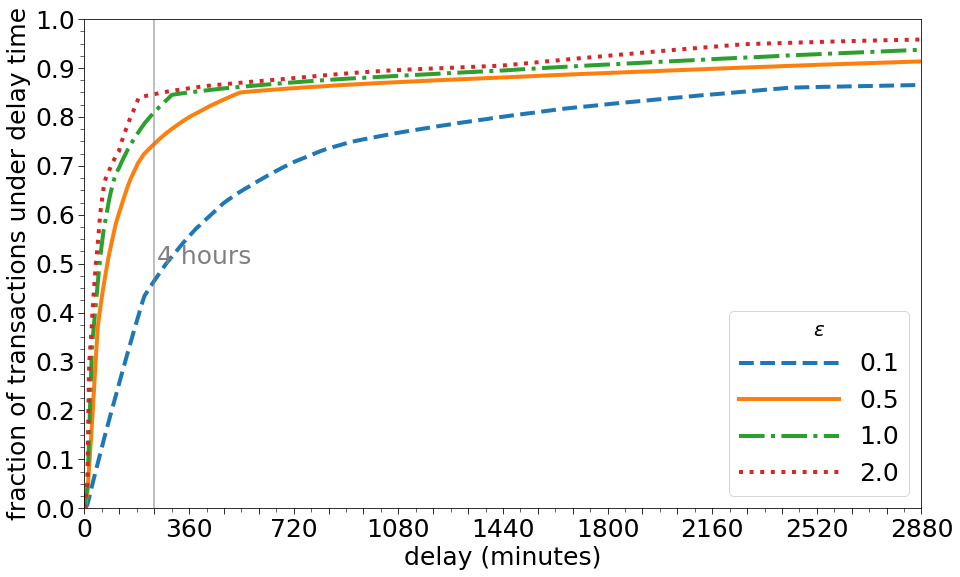}
    \caption{Cumulative distribution of delay added to batched Bitcoin transactions (averaged over 5 trials). Delay is drawn from a privacy-preserving uniform distribution with $\gap$ set to the median of the inter-arrival time of transactions to an output address within the past $7$ days.}
    \label{fig:btc_delay_75}
\end{figure}

\end{appendix}

\end{document}